\documentclass{article}

\usepackage{makeidx}
\usepackage{latexsym}
\usepackage{amsfonts}
\usepackage{amssymb}
\usepackage{amsmath}
\usepackage{amstext}
\usepackage{amsthm}
\usepackage{mathrsfs}
\usepackage{geometry}
\usepackage{graphicx}
\usepackage{mathabx}
\date{}

\newtheorem{theorem}{Theorem}
\newtheorem{example}[theorem]{Example}

\newtheorem{lemma}[theorem]{Lemma}

\newtheorem{corollary}[theorem]{Corollary}

\newtheorem{proposition}[theorem]{Proposition}

\begin{document}

\title{The flow network method\footnote{We wish to thank an anonymous referee
for letting us know the existence of the papers by Gvozdik (1987) and Belkin and Gvozdik (1989) (in Russian), where the flow network method was first formulated. We also thank Andrey Sarychev for translating the mentioned papers.
Daniela Bubboloni was partially supported by GNSAGA of INdAM.}}
\author{\textbf{Daniela Bubboloni} \\
%EndAName
{\small {Dipartimento di Matematica e Informatica U.Dini} }\\
\vspace{-6mm}\\
{\small {Universit\`{a} degli Studi di Firenze} }\\
\vspace{-6mm}\\
{\small {viale Morgagni 67/a, 50134 Firenze, Italy}}\\
\vspace{-6mm}\\
{\small {e-mail: daniela.bubboloni@unifi.it}}\\
\vspace{-6mm}\\
{\small tel: +39 055 2759667} \and \textbf{Michele Gori}
 \\
%EndAName
{\small {Dipartimento di Scienze per l'Economia e  l'Impresa} }\\
\vspace{-6mm}\\
{\small {Universit\`{a} degli Studi di Firenze} }\\
\vspace{-6mm}\\
{\small {via delle Pandette 9, 50127 Firenze, Italy}}\\
\vspace{-6mm}\\
{\small {e-mail: michele.gori@unifi.it}}\\
\vspace{-6mm}\\
{\small tel: +39 055 2759707}}

\maketitle

\begin{abstract}
\noindent   In this paper we propose an in-depth analysis of a method, called the flow network method, which associates with any network a complete and quasi-transitive binary relation on its vertices. Such a method, originally proposed by Gvozdik (1987), is based on the concept of maximum flow.
Given a competition involving two or more teams, the flow network method can be used to build a relation on the set of teams which establishes, for every ordered pair of teams, if the first one did at least as good as the second one in the competition. Such a relation
naturally induces procedures for ranking teams and selecting the best $k$ teams of a competition. Those procedures are proved to satisfy many desirable properties.
\end{abstract}
\vspace{4mm}

\noindent \textbf{Keywords:} \noindent Network, network flow, quasi-transitive relation, pairs comparison, ranking problem, Schulze method.

\vspace{2mm}



\section{Introduction}

The Lega Basket Serie A, the first and highest-tier level of the Italian basketball league, currently consists of 16 teams. During the regular season of the basketball championship every team plays every other team twice for a total of 30 matches and ties are not allowed. At the end of the
season, teams are ranked by the total number of wins. Basketball fans certainly like having a ranking of the teams available at any stage of the season.
The usual method to build those temporary rankings is still based on counting the number of wins of each team. However, it is simple to understand that such a method may
lead to questionable outcomes because, even though two teams won the same number of matches at a given moment, they might have played their matches against opponents of different quality. That suggests the need to look for more expressive ranking methods.

\vspace{3mm}

Consider a set of teams and a list of matches played by those teams, where ties were not allowed. Such matches are supposed to be all the matches scheduled in a (possibly fictitious) competition among the considered teams. In order to rank the teams and select the winners one may compare them on the basis of the number of wins and, in fact, this method is largely used when each team confronts any other team the same number of times, as in the case of round-robin competitions. However, there are interesting situations which do not meet such a condition as the case of temporary rankings previously discussed or when the number of teams involved in a competition is so large that, for practical reasons, it is difficult to arrange all the possible matches\footnote{The National Football League (NFL) and the NCAA Division I Football Bowl Subdivision are examples of competitions where there are teams which never play each other on the field.}. In these cases, counting the number of wins does not seem to be such a good idea.
Consider, for instance,
the competition among the three teams $\textsc{a}$, $\textsc{b}$ and $\textsc{c}$ described by the table\footnote{We describe a competition by a table in which every row has the shape
\[
\begin{array}{|c||c|c||c|c|}
\hline
x & m &  n  &y\\
\hline
\end{array}
\]
meaning that the matches involving the teams $x$ and $y$ were $m+n$ and that $x$ won $m$ times and $y$ won $n$ times.}
\begin{equation}\label{table1}
\begin{array}{|c||c|c||c|c|}
\hline
\textsc{a} & 2 & 0  &\textsc{b}\\
\hline
\textsc{a} & 2 & 0 &\textsc{c}\\
\hline
\textsc{b} & 5 & 0 &\textsc{c}\\
\hline
\end{array}
\end{equation}
and note that not all the pairs of teams confronted each other the same number of times. Here, counting the number of wins makes team
$\textsc{b}$ be the winner, even though most people probably believe that team $\textsc{a}$ deserves to win the competition.

In order to overcome those  difficulties, Gvozdik (1987) and Belkin and Gvozdik (1989) propose a method, that we are going to call the {\it flow network method},
which can be used to establish, for every competition and every pair of teams involved in the competition, if the first team did at least as good as the second one. In other words, the method associates with every competition a special binary relation on the set of teams. Moreover, it naturally allows to select the winners of the competition as well as to determine the rankings of the teams that are consistent with the competition\footnote{We stress that Belkin and Gvozdik (1989) mainly focus on the problem of building rankings.}. We independently rediscover the flow network method and we strongly believe that it is worthwhile to
provide an in-depth analysis of it by revisiting old results and presenting new ones. Indeed, since the two above mentioned papers are published only in Russian, the flow network method has not certainly received the attention it deserves yet.

First of all, let us describe how the flow network method operates through the analysis of a concrete example. Consider the competition \textsf{C} among the four teams $\textsc{a}$, $\textsc{b}$, $\textsc{c}$ and $\textsc{d}$ described by the following table:
\begin{equation}\label{table2}
\begin{array}{|c||c|c||c|c|}
\hline
\textsc{a} & 1 & 0  &\textsc{b}\\
\hline
\textsc{a} & 1 & 2 &\textsc{c}\\
\hline
\textsc{a} & 2 & 2 &\textsc{d}\\
\hline
\textsc{b} & 1 & 2 &\textsc{c}\\
\hline
\textsc{b} & 1 & 1 &\textsc{d}\\
\hline
\textsc{c} & 2 & 2   &\textsc{d}\\
\hline
\end{array}
\end{equation}
Given two distinct teams $x$ and $y$, let us use the writing $x y$ to denote a match between $x$ and $y$ where $x$ beat $y$. Thus, looking at the table we can list all the matches played in \textsf{C} and the corresponding winners as follows
\begin{equation}\label{list1}
\begin{array}{l}
\textsc{ab},
\textsc{ac},\textsc{ca},\textsc{ca},
\textsc{ad}, \textsc{ad}, \textsc{da}, \textsc{da},
\textsc{bc}, \textsc{cb}, \textsc{cb},
\textsc{bd}, \textsc{db},
\textsc{cd}, \textsc{cd}, \textsc{dc}, \textsc{dc}.
\end{array}
\end{equation}
Fix now two distinct teams $x$ and $y$. Let us call path from $x$ to $y$ (in \textsf{C}) any sequence $x_1\cdots x_n$ of $n\ge 2$ distinct teams $x_1,\ldots,x_n$ such that $x_1=x$, $x_n=y$ and $x_1 x_2,\ldots, x_{n-1}x_n$ are all in \eqref{list1}.
For instance, $\textsc{ab}$ is a path  from $\textsc{a}$ to $\textsc{b}$ since $\textsc{ab}$ is in \eqref{list1}; $\textsc{acd}$ is a path from $\textsc{a}$ to
$\textsc{d}$ since $\textsc{ac}$ and $\textsc{cd}$ are in \eqref{list1}; $\textsc{bac}$ is not a path from $\textsc{b}$  to $\textsc{c}$ since $\textsc{ba}$ is not in  \eqref{list1}.
Given a path $x_1\cdots x_n$ from $x$ to $y$, we identify it with the sub-competition of \textsf {C} involving only the teams $x_1,\ldots, x_n$ and where only the matches $x_1 x_2,\ldots, x_{n-1}x_n$ were played. Of course, the intuition suggests that $x$ did better than $y$ in that sub-competition. For instance, we obviously have that $\textsc{a}$ did better than $\textsc{b}$ in $\textsc{ab}$. Moreover, it is natural to state that $\textsc{a}$ did better than $\textsc{d}$ in $\textsc{acd}$ since in $\textsc{acd}$ only two matches were played, namely $\textsc{a}$ against $\textsc{c}$ and $\textsc{c}$ against $\textsc{d}$, and we know that
$\textsc{a}$ beat $\textsc{c}$ and $\textsc{c}$ beat $\textsc{d}$ so that, in some sense, $\textsc{a}$  indirectly beat $\textsc{d}$.

Denote next by $\lambda_{xy}$ the maximum length of a list of paths from $x$ to $y$ that can be built using each element in \eqref{list1} at most once. Note that the same path can appear more than once in a list. In order to clarify the definition of such a number, let us consider $\textsc{a}$ and $\textsc{b}$ and compute $\lambda_{\textsc{ab}}$. First of all, note that
\[
\textsc{ab}, \textsc{adcb}, \textsc{adcb},  \textsc{acdb},
\]
are four paths from $\textsc{a}$ to $\textsc{b}$  in \textsf {C} which are built using the following matches in \eqref{list1}
\[
\textsc{ab}, \textsc{ad}, \textsc{dc}, \textsc{cb}, \textsc{ad}, \textsc{dc}, \textsc{cb}, \textsc{ac}, \textsc{cd}, \textsc{db}.
\]
The matches in \eqref{list1} which have not been used yet are
\[
\textsc{ca}, \textsc{ca}, \textsc{da}, \textsc{da}, \textsc{bc}, \textsc{bd}, \textsc{cd},
\]
and it is clear that no further path from $\textsc{a}$ to $\textsc{b}$  can be built using them. Of course, matches in \eqref{list1} might be arranged to get a different list of paths from $\textsc{a}$ to $\textsc{b}$. For instance, we have that
\[
\textsc{ab}, \textsc{acb}, \textsc{adcb}, \textsc{adb},
\]
is a different family of four paths from $\textsc{a}$ to $\textsc{b}$  which are built using the following matches in \eqref{list1}
\[
\textsc{ab}, \textsc{ac}, \textsc{cb}, \textsc{ad},  \textsc{dc},  \textsc{cb}, \textsc{ad}, \textsc{db}.
\]
Also in this case, no further path from $\textsc{a}$ to $\textsc{b}$ can be built using the matches left out.
Observe now that it is not possible to find more than four paths from $\textsc{a}$ to $\textsc{b}$. Indeed,  in \eqref{list1} the number of matches of the type $x \textsc{b}$ where $x\in\{\textsc{a},\textsc{c},\textsc{d}\}$ is four and any path from $\textsc{a}$ to $\textsc{b}$ has to involve exactly one match of that type. As a consequence, we get $\lambda_{\textsc{ab}}=4$.

As a further example, let us now compute $\lambda_{\textsc{ba}}$.
First, observe that  $\textsc{bca}$ and $\textsc{bda}$
are two paths from $\textsc{b}$ to $\textsc{a}$  which are built using the matches $\textsc{bc}$, $\textsc{ca}$, $\textsc{bd}$ and $\textsc{da}$ in \eqref{list1}.
Moreover, we cannot build more than two paths from $\textsc{b}$ to $\textsc{a}$ since the number of matches of the type $\textsc{b}x$ where $x\in\{\textsc{a},\textsc{c},\textsc{d}\}$ is two and any path from $\textsc{b}$ to $\textsc{a}$ must involve exactly one match of that type. Thus, we get $\lambda_{\textsc{ba}}=2$.

Using similar strategies, one can easily compute, for every pair of distinct teams $x$ and $y$, the number $\lambda_{xy}$. Such computations give
\begin{equation*}
\begin{array}{l}
\lambda_{\textsc{ab}}=4,\;\lambda_{\textsc{ba}}=2,\;\lambda_{\textsc{ac}}=4,\; \lambda_{\textsc{ca}}=4,\; \lambda_{\textsc{ad}}=4,\;
\lambda_{\textsc{da}}=4,\; \\
\lambda_{\textsc{bc}}=2,\; \lambda_{\textsc{cb}}=4,\; \lambda_{\textsc{bd}}=2,\; \lambda_{\textsc{db}}=4,\;\lambda_{\textsc{cd}}=5,\;
 \lambda_{\textsc{dc}}=4.
\end{array}
\end{equation*}
Given two distinct teams $x$ and $y$,
we interpret $\lambda_{xy}$ as the number of times the team $x$ directly or indirectly beat the team $y$ in the competition $\mathsf{C}$\footnote{A similar interpretation appears in Patel (2015).}.
As a consequence, for any pair of distinct teams $x$ and $y$, we interpret the inequality $\lambda_{xy}\ge \lambda_{yx}$ as the fact that $x$ did at least as good as $y$ in $\mathsf{C}$. In that case, we write $x\succeq_\textsf {C} y$. Of course, since we can assume that each team did at least as good as itself in $\mathsf{C}$, we also set $x\succeq_{\mathsf{C}} x$ for all
$x\in \{\textsc{a},\textsc{b},\textsc{c},\textsc{d}\}$. Thus, we get a relation on the set of teams which is fully described as follows
\begin{equation}\label{c-intro}
\begin{array}{l}
\textsc{a}\succeq_{\textsf {C}} \textsc{a},\;
\textsc{b}\succeq_{\textsf {C}} \textsc{b},\;
\textsc{c}\succeq_{\textsf {C}} \textsc{c},\;
\textsc{d}\succeq_{\textsf {C}} \textsc{d},\;\\
\textsc{a}\succeq_{\textsf {C}} \textsc{b},\; \textsc{a}\succeq_{\textsf {C}}\textsc{c}, \; \textsc{c}\succeq_{\textsf {C}} \textsc{a}, \;
\textsc{a}\succeq_{\textsf {C}} \textsc{d},\\
\textsc{d}\succeq_{\textsf {C}}\textsc{a},\;\textsc{c}\succeq_{\textsf {C}} \textsc{b}, \; \textsc{d}\succeq_{\textsf {C}} \textsc{b}, \;
\textsc{c}\succeq_{\textsf {C}} \textsc{d}.
\end{array}
\end{equation}
We call  $\succeq_{\textsf{C}}$ the flow relation associated with $\mathsf{C}$. Because of its definition, we interpret its maxima, namely $\textsc{a}$ and $\textsc{c}$, as the winners of the competition $\mathsf{C}$ and its linear refinements, namely $\textsc{a}\succeq_{\textsf {C}}\textsc{c}\succeq_{\textsf {C}}\textsc{d}\succeq_{\textsf {C}}\textsc{b}$,
$\textsc{c}\succeq_{\textsf {C}}\textsc{a}\succeq_{\textsf {C}}\textsc{d}\succeq_{\textsf {C}}\textsc{b}$ and $\textsc{c}\succeq_{\textsf {C}}\textsc{d}\succeq_{\textsf {C}}\textsc{a}\succeq_{\textsf {C}}\textsc{b}$ as the rankings of the teams which are consistent with the competition $\mathsf{C}$.
We stress that $\succeq_{\mathsf{C}}$ is complete but it is not transitive because, for instance, we have that $\textsc{d}\succeq_{\textsf {C}} \textsc{a}$
and $\textsc{a}\succeq_{\mathsf {C}} \textsc{c}$ but $\textsc{d}\not\succeq_{\textsf {C}} \textsc{c}$.

Observe also that, as it is immediately checked, the flow relation associated with the competition described by \eqref{table1} is a linear order having $\textsc{a}$ ranked first, $\textsc{b}$ ranked second and $\textsc{c}$ ranked third. In particular, the unique winner is $\textsc{a}$ and the unique admissible ranking is the flow relation itself, as one would desire.

It is also implicit in its definition that the flow network method takes into account the quality of the opponents. Indeed, if in a competition a team $\textsc{a}$ won against teams of good quality, namely which won many matches, then there exist many paths from $\textsc{a}$ to the other teams.
That potentially makes $\textsc{a}$ behave better than other teams in the flow relation.

\vspace{3mm}

The flow network method owes its name to the fact that networks and flows are the basic concepts underlying its definition. In order to explain this fact, first note that there is a natural bijection between the set of competitions and the set of networks. Indeed, any competition can be identified with the network whose vertices are the teams involved in the competition and where, for every pair of distinct teams, the capacity of the arc from a team to another one is the number of times the first team beat the second one.
Similarly, any network can be thought as a competition among its vertices where, for every arc, its capacity represents how many times its start vertex beat its end vertex\footnote{Due to the identification between competitions and networks, we are going to freely use the terminology of competitions for networks too.}. Moreover, given a  network,  it can be proved that, for every pair of distinct  vertices $x$ and $y$, the number
$\lambda_{xy}$ previously described equals the so-called maximum flow value from $x$ to $y$.
In the paper, after having introduced in Section \ref{sec-network} suitable notation for networks and some crucial results about them, we give in Section \ref{def-flow} the formal definition of the flow network method. Then, in order to assess the method, in Section \ref{flow-network-method} we deepen the analysis  carried on by
Belkin and Gvozdik (1989). We first show that the outcomes of the flow network method are always complete and quasi-transitive relations (Theorem \ref{F-qt}),
even though not necessarily transitive (Proposition \ref{non-o}). We show then that the method may lead to any possible complete and quasi-transitive relation
(Proposition \ref{main-qt}) and that it is neutral, homogeneous, efficient, monotonic, reversal symmetric and symmetric (Propositions \ref{neut-F}, \ref{homo-net}, \ref{Eff-F}, \ref{Mon1-flow-rule}, \ref{Reversal}, \ref{flat}). Remarkably, we also obtain a characterization of complete and quasi-transitive relations (Theorem \ref{co-qt}).

Networks are actually used to model a variety of situations, not necessarily related to sport competitions, where suitable information about pairwise comparisons among different alternatives is known\footnote{Note that in many applications networks capacities are allowed to be nonnegative real numbers and networks are identified with their adjacency matrix.} and a ranking of the alternatives is needed. Thus, it is certainly possible to find in the literature lots of network methods, that is, procedures to associate with any network a binary relation (usually complete and transitive) on the set of teams describing when a team is at least as good as another one. Those methods have been designed in response to very different needs so that the rationale behind their definition, as well as the properties they fulfil, strongly depend on the specific framework they are supposed to manage. Langville and Meyer (2012) present an overview of different contexts and settings (like social choice, voting,  web search engines, psychology and statistics) where ranking alternatives is important and describe lots of ranking procedures that can be formalized as network methods.

Among the simplest network methods there are surely
the Borda method, which associates with each network the relation which compares teams on the basis of the whole number of wins;
the dual Borda method, which compares teams on the basis of the whole number of losses;
the Copeland method\footnote{The Copeland method is sometimes called net flow method (Bouyssou, 1992).}, which compares teams on the basis of the difference between the whole number of wins and the whole number of losses;
the minimax method which compares teams on the basis of the maximum number of losses against each other team;
the maxmin method which compares teams on the basis of the minimum number of wins against each other team.
Gonzal\'ez-D\'iaz et al. (2014) consider several network methods originating from statistical and social choice approaches and compare them by studying how those methods perform with respect to a specific set of properties\footnote{We emphasize that some of the properties considered by Gonzal\'ez-D\'iaz et al. (2014) are satisfied by the flow network method as described in course of the paper.}.
Vaziri et al. (2017) focus on some popular sports network methods and discuss whether they satisfy suitable properties that a fair and comprehensive ranking method should
meet\footnote{The flow network method fulfils the three main properties stated in that paper. Property I (opponent strength) has been already discussed; Property II (incentive to win) is the content of Proposition \ref{Mon1-flow-rule}; Property III (sequence of matches) follows from the very definition of the method.}.
Several authors deepen the analysis of the properties of networks methods providing an axiomatic characterization of some of them.
Bouyssou (1992), Palacios-Huerta and Volij (2004), Slutzki and Volij (2005), van den Brink and Gilles (2009) and Kitti (2016) respectively
show that the Copeland, the normalized invariant, the fair-bets, the Borda and the non-normalized invariant network methods can be characterized by suitable sets of axioms\footnote{Note that the (normalized) invariant and the fair-bets methods are also known as the (normalized) long path and the Markov methods, respectively.}.
Csat\'o (2017) proves instead an impossibility result showing the incompatibility of the self-consistency and the order preservation principles.
Further, Laslier (1997, Chapters 3 and 10) presents an extensive survey of various scoring methods defined on the set of tournaments\footnote{A tournaments is a complete and asymmetric digraph. Tournaments can be identified with networks whose capacities are 0 or 1 and such that the sum of the capacities of any pair of opposite arcs is 1. They are used to represent round robin-competitions.} and on the set of balanced networks\footnote{A balanced network is a network whose sum of the capacities of any pair of opposite arcs is constant. Balanced networks represent competitions where any pair of teams confront each other the same number of time.}.

All the above mentioned methods, as the large majority of methods available in the literature, are based on suitable definitions of scores for teams so that they
determine a complete and transitive relation on the set of teams. An exception is the network method which can be deduced from the paper by Schulze (2011) and which is described and studied in Section \ref{schulze-sec}. In fact, as the flow network method, such a network method has the peculiar property to associate with any network a complete and quasi-transitive relation on its vertices which is not in general transitive. As a consequence, the flow network method is different from all those methods leading to complete and transitive relations and, as proved in Section \ref{schulze-sec}, it is also different from the Schulze network method. However, we show that the flow and the Borda network methods agree on the class of balanced networks (Proposition \ref{balan-ut})\footnote{The equality between the flow and the Borda network methods on balanced networks is stated, without proof, in  Belkin and Gvozdik (1989, Lemma 1).}. This is a remarkable fact as it shows that the flow network method can be seen as an extension of the well-known method of ``counting the wins'' outside the ordinary framework where it is generally accepted. But we discovered other classes of networks on which those methods agree (Proposition \ref{cor-new} and Theorem \ref{OI}) and that shed light on the significance of the Borda network method outside the traditional context of the balanced networks.

As a final remark, we recall that the computation of flows in a network can be performed by well-known algorithms, based on the Ford and Fulkerson augmenting paths algorithm, which run in polynomial time with respect to the number of vertices. As a consequence, the outcome of the flow network method applied to a competition can be computed in a polynomial time with respect to the number of teams.
That fact is certainly encouraging with a view to concretely applying the method.

\vspace{3mm}

After having carefully studied the properties of the flow network method,
we focus on the procedure which associates with any network the family of the linear refinements of the flow relation\footnote{Complete and quasi-acyclic relations always admit linear refinements (Proposition \ref{lin-ref}).}.
Such a procedure, formally defined in Section \ref{def-flow}, is called {\it flow network rule} and was first defined in Belkin and Gvozdik (1989).  On the basis of the discussion made about the flow network method, we have that it is in general different from the network rules built by the known network methods by taking, for every network, the linear refinements of the corresponding relation. It also generally differs from  the Kemeny and the ranked pairs rules\footnote{That fact can be easily checked considering the network in Example \ref{flow-schulze-1}.}. In Section \ref{fnr} we prove several properties of the flow network rule.

We focus then on the problem of selecting a subset of a given size from the set of vertices of a network. Mainly in the framework of voting theory
the problem of selecting committees has been receiving more and more attention in the last years (for a survey, see Elkind et al., 2017). Many $k$-multiwinner network solutions, that is procedures which associate with any network a family of subsets of size $k$, can be built by each complete and quasi-transitive (in particular transitive) network method by taking, for every network, the $k$-maximum sets of the corresponding relation\footnote{Given a relation $R$ on a nonempty finite set $V$, a $k$-maximum set of $R$ is a subset $W$ of size $k$ of $V$ having the property that, for every $x\in W$ and $y\not\in W$, $(x,y)\in R$. Complete and quasi-acyclic relations always admit $k$-maximum sets (Proposition \ref{ckR}).}. A careful analysis of such solutions as well as a comparison among them represent, in our opinion, an interesting research line.
However, at the best of our knowledge, no specific contribution on this theme is available in the literature when $k\ge 2$. The situation is different for the case $k=1$.
Aziz et al. (2015) propose a canonical way to extend to the whole set of networks any method which associates with each balanced network a subset of its vertices.
There are also lots of contributions about procedures to select the vertices of
special types of networks like tournaments, weak tournaments and partial tournaments\footnote{Weak tournaments and partial tournaments can be naturally identified with networks whose arcs always have capacities that are  0 or 1 and such that the sum of the capacities of each pair of opposite arcs is at least 1 and at most 1 respectively.}.
We refer to Laslier (1997) for a survey on tournament solutions and to Peris and Subiza (1999),
Aziz et al. (2015) and  Brandt et al. (2016) for further interesting contributions on the extension of classical tournament solutions to the settings of weak and partial tournaments.
We finally note that De Donder et al. (2000) also propose some results about balanced networks, while
Dutta and Laslier (1999) consider a different framework somehow connected with networks, as discussed in Section \ref{dutta-sec}.
In Section \ref{def-flow} we consider,
for every positive integer $k$, the procedure which associates with any network the $k$-maximum sets of its flow relation. Such a procedure, called {\it flow $k$-multiwinner network solution}, is a new object and is in general different from $k$-multiwinner network solutions originating from known network methods and network rules.
In Sections \ref{fns}, thanks to the previous analysis of the flow network method as well as some preliminary propositions about quasi-acyclic and quasi-transitive relations proved in Section \ref{relation-2}, we show that the flow $k$-multiwinner network solution fulfils a lot of desirable properties.

\section{Networks}\label{sec-network}
We assume $0\not\in \mathbb{N}$ and we set $\mathbb{N}_0=\mathbb{N}\cup\{0\}$.
Moreover, given  $k\ge 1$, we denote by $\ldbrack k \rdbrack$ the set $\{1,\ldots,k\}$.
From now on, $V$ is a fixed finite set of size $|V|=n\ge 2$, $\Delta=\{(x,x)\in V^2:x\in V\}$ and $A=V^2\setminus \Delta$.

A  {\it network} on $V$ is a triple $N=(V,A,c)$, where $c$ is a function from $A$ to $\mathbb{N}_0$.
We say that $V$ is the set of vertices of $N$, $A$ is the set of arcs of $N$ and $c$ is the capacity associated with $N$.
Note that the pair $(V,A)$ is a complete  digraph (without loops) on the set of vertices $V$. The set of networks on $V$ is denoted by $\mathcal{N}$.

Networks can be used to mathematically represent competitions. Assume to have a set of teams which played a certain number of matches among each other and to know, for every team, the number of matches it won against any other team. Then we can represent that competition by a network by defining $V$ as the set of teams and, for every $(x,y)\in A$, defining $c(x,y)$ as the number of matches in which $x$ beat $y$.
For instance, the competition $\mathsf{C}$ described in \eqref{table2} can be represented by the network $N_{\mathsf{C}}=(V,A,c)$ where
\begin{equation}\label{net-comp}
\begin{array}{l}
V=\{\textsc{a},\textsc{b},\textsc{c},\textsc{d}\};\\
\vspace{-2mm}\\
A=\{(x,y)\in V^2: x,y\in\{\textsc{a},\textsc{b},\textsc{c},\textsc{d}\} \mbox{ and } x\neq y\};\\
\vspace{-2mm}\\
c:V\to\mathbb{N}_0 \mbox{ is defined by }\\
c(\textsc{a},\textsc{b})=1,\; c(\textsc{b},\textsc{a})=0,\;c(\textsc{a},\textsc{c})=1,\; c(\textsc{c},\textsc{a})=2,\\
c(\textsc{a},\textsc{d})=2,\; c(\textsc{d},\textsc{a})=2,\;  c(\textsc{b},\textsc{c})=1,\;c(\textsc{c},\textsc{b})=2,\\
c(\textsc{b},\textsc{d})=1,\; c(\textsc{d},\textsc{b})=1,\;c(\textsc{c},\textsc{d})=2,\; c(\textsc{d},\textsc{c})=2.\\
\end{array}
\end{equation}

Let $N=(V,A,c)\in\mathcal{N}$. For every $x\in V$, the outdegree and the indegree of $x$ in $N$ are respectively defined by
\[
o(x)=\sum_{y\in V\setminus\{x\}}c(x,y),
\qquad
i(x)=\sum_{y\in V\setminus\{x\}}c(y,x).
\]
Note that $o(x)$ and $i(x)$ depend on $N$ but the notation $o(x)$ and $i(x)$ does not refer to the symbol $N$. When the reference to the network is important we adopt the writings $o^N(x)$ and $i^N(x).$ We use a similar style for any  other symbol defined for networks.

We call reversal of $N$ the network $N^r=(V,A,c^r)\in\mathcal{N}$ where, for every $(x,y)\in A$, $c^r(x,y)=c(y,x)$. Given a bijection $\psi:V\rightarrow V$, we set $N^{\psi}=(V,A,c^{\psi}) \in\mathcal{N}$ where $c^{\psi}$ is defined, for every $(x,y)\in A$, by
$c^{\psi}(x,y)=c(\psi^{-1}(x), \psi^{-1}(y))$. Moreover, given $\alpha\in\mathbb{N}$, we set $\alpha N=(V,A,\alpha c)\in \mathcal{N}$.

Given $k\in \mathbb{N}_0$, we say that $N$ is a {\it $k$-balanced network} if, for every $(x,y)\in A$, $c(x,y)+c(y,x)=k$. In that case $k$ is called the balance of $N$.
We denote the set of $k$-balanced networks on $V$ by $\mathcal{B}_k$ and we call $\mathcal{B}=\bigcup_{k\in \mathbb{N}_0}\mathcal{B}_k$  the set of balanced networks on $V$. Note that modelling competitions where each team confronts with any other team the same number of times leads to balanced networks.
In particular,  networks in $\mathcal{B}_1$ are the right tool to model round-robin competitions, namely those competitions where each team confronts with any other team exactly once. Those special competitions are largely studied in the literature and are generally modelled via asymmetric and complete digraphs, the so-called tournaments. Of course,  there is a natural bijection among the set of tournaments on $V$ and the set $\mathcal{B}_1$.

Let $N=(V,A,c)\in \mathcal{N}$ and $s,t\in V$ with $s\neq t$. A {\it flow} from $s$ to $t$ in $N$ is a function $f:A\to \mathbb{N}_0$ such that, for every $(x,y)\in A$, $f(x,y)\le c(x,y)$ and, for every $x\in V\setminus\{s,t\}$,
\begin{equation}\label{conservation}
\sum_{y\in V\setminus\{x\}}f(x,y)=\sum_{y\in V\setminus\{x\}}f(y,x).
\end{equation}
Note that the function $f_0:A\to \mathbb{N}_0$ defined by $f_0(x,y)=0$ for all $(x,y)\in A$ is a flow. It is called the null flow.
The set of flows from $s$ to $t$ in $N$ is finite and it is denoted by $\mathcal{F}(N,s,t)$.
Given $f\in \mathcal{F}(N,s,t)$,  the value of $f$ is the integer
\[
\varphi(f)=\sum_{y\in V\setminus\{s\}}f(s,y)-\sum_{y\in V\setminus\{s\}}f(y,s).
\]
The number
\[
\varphi_{st}=\max_{f\in \mathcal{F}(N, s,t)} \varphi(f),
\]
which is well defined and belongs to $\mathbb{N}_0$, is called the {\it maximum flow value} from $s$ to $t$ in $N$.
If $f\in \mathcal{F}(N,s,t)$ is such that $\varphi(f)=\varphi_{st}$, then $f$ is called a maximum flow from $s$ to $t$ in $N$.
Let $S\subseteq V$ and denote by $S^c$ its complement in $V$. If $\varnothing \neq S\neq V$, the number
\[
c(S)=\sum_{x\in S,\, y\in S^c} c(x,y)
\]
is called the capacity of $S$ in $N$. $S$ is called a {\it cut} from $s$ to $t$ (or for $\varphi_{st}$) in $N$ if $s\in S$ and $t\in S^c$.

The set of cuts from $s$ to $t$ in $N$ is nonempty and finite and it is denoted by $\mathcal{C}(N,s,t)$.
It is well-known that, for every $f\in \mathcal{F}(N,s,t)$ and $S\in \mathcal{C}(N,s,t)$, we have
\begin{equation}\label{duality}
\varphi(f)\le c(S).
\end{equation}
In particular, if $S_1=\{s\}$ and $S_2=V\setminus\{t\}$, then $S_1,S_2\in \mathcal{C}(N,s,t)$, $c(S_1)=o(s)$ and $c(S_2)=i(t).$ Thus, by
\eqref{duality}, we get $\varphi_{st}\le \min\{o(s),i(t)\}.$

The fundamental result about flows in networks, namely the famous Maxflow-Mincut Theorem, states that
\[
\varphi_{st}=\min_{S\in \mathcal{C}(N,s,t)} c(S).
\]
In particular, there always exists $S\in \mathcal{C}(N,s,t)$ such that $\varphi_{st}=c(S)$. Such a cut is called a {\it minimum cut} from $s$ to $t$ in $N$ (or for $\varphi_{st}$).
The next result will turn out to be fundamental for our work. It is due to Gomory and Hu (1961).

\begin{proposition}\label{gomory}
Let $N\in \mathcal{N}$ and $x,y, z\in V$ be distinct. Then
$\varphi_{xz}\ge\min\{\varphi_{xy},\varphi_{yz}\}.$
\end{proposition}

\begin{proof}Let $S$ be a minimum cut from $x$ to $z$ in $N$. Then we have $x\in S$, $z\in S^c$ and $\varphi_{xz}=c(S)$.
 If $y\in S$, then $S$ is a cut from $y$ to $z$ and thus, by \eqref{duality}, $\varphi_{yz}\leq c(S)$.
If instead $y\in S^c$ then $S$ is a cut from $x$ to $y$ and thus, by \eqref{duality}, $\varphi_{xy}\leq c(S)$. In both cases we get $\min\{\varphi_{xy},\varphi_{yz}\}\leq \varphi_{xz}.$ \end{proof}

We call a sequence of $m\ge 2$ distinct vertices $x_1,\ldots, x_m$ such that $x_1=s$, $x_m=t$ and, for every $i\in\ldbrack m-1 \rdbrack$, $c(x_i,x_{i+1})\ge 1,$ a path from $s$ to $t$ in $N$ and we denote it by
$x_1\cdots x_m$. The set of paths from $s$ to $t$ in $N$ is denoted by $\Gamma(N,s,t)$. Given $\gamma=x_1\cdots x_m\in \Gamma (N,s,t)$, we set $V(\gamma)=\{x_1,\ldots,x_m\}$ and $A(\gamma)=\{(x_1,x_2), \ldots, (x_{m-1},x_m)\}$.
Fixed $k\in \mathbb{N}$, we define an arc-disjoint $k$-sequence in $\Gamma(N,s,t)$ as a sequence $(\gamma_{j})_{j=1}^k$ of $k$ (not necessarily distinct) paths in $\Gamma(N,s,t)$
such that, for every $a\in A$,
$|\{j\in\ldbrack k \rdbrack: a\in A(\gamma_j)\}|\le c(a).$
We denote the set of arc-disjoint $k$-sequences in $\Gamma(N,s,t)$ by $\Gamma_k(N,s,t)$. Of course, if $\Gamma(N,s,t)=\varnothing$, then $\Gamma_k(N,s,t)=\varnothing$ for all $k\in \mathbb{N}$.
We define \[\lambda_{st}=\left\{
\begin{array}{ll}
\max\{k\in \mathbb{N}:  \Gamma_k(N,s,t)\neq\varnothing\}&\mbox{ if } \Gamma(N,s,t)\neq\varnothing\\
\vspace{-2mm}\\
0 &\mbox{ if } \Gamma(N,s,t)=\varnothing.\\
\end{array}
\right.
\]
From Bang-Jensen and Gutin (2008, Lemma 7.1.5), we have the equality
\begin{equation}\label{JG}
\varphi_{st}=\lambda_{st},
\end{equation}
which  provides a useful interpretation of the maximum flow value in terms of arc-disjoint paths. Note that, due to \eqref{JG}, the maximum flow value from $x$ to $y$ in the network $N_{\mathsf{C}}$  described in \eqref{net-comp} has been computed in the introduction for all distinct $x,y\in V$.

\section{Relations}

A relation $R$ on $V$ is a subset of $V^2$. It is customary to use relations on $V$ to represent individual or social preferences on $V$ by identifying, for every $x,y\in V$, the membership relation $(x,y)\in R$ with the statement ``$x$ is at least as good as $y$''.
We denote the set of relations on $V$ by $\mathbf{R}$.
	
Let $R\in\mathbf{R}$. The strict relation associated with $R$ is $S(R)=\{(x,y)\in V^2:(x,y)\in R,(y,x)\notin R\}$;
the reversal relation of $R$  is $R^r=\{(x,y)\in V^2:(y,x)\in R\}.$
Given $x,y\in V$, we usually write $x\succeq_R y$ instead of $(x,y)\in R$;
$x\succ_R y$ instead of $(x,y)\in S(R)$; $x\sim_R y$  instead of $(x,y)\in R$ and $(y,x)\in R$.
Recall that, by definition, $R$ is quasi-transitive if $S(R)$ is transitive;
acyclic if, for every sequence $x_1,\ldots,x_m$ of $m\ge 2$ distinct elements of $V$ such that $x_i\succeq_R x_{i+1}$ for all $i\in\ldbrack m-1 \rdbrack$, we have that $x_m \not\succeq_R x_1$; quasi-acyclic if $S(R)$ is acyclic. Given a function $u:V\to \mathbb{R}$, the relation induced by $u$ is $R(u)=\{(x,y)\in V^2: u(x)\geq u(y)\}$.

We denote the set of complete and quasi-acyclic relations on $V$ by $\mathbf{A}$;
the set of complete and quasi-transitive relations on $V$ by $\mathbf{T}$;
the set of complete and transitive relations on $V$ by $\mathbf{O}$;
the set of complete, quasi-transitive and antisymmetric relations on $V$ by $\mathbf{L}$.
The elements of $\mathbf{O}$ are called orders; the ones of $\mathbf{L}$ are called linear orders. Of course, we have $\mathbf{L}\subseteq \mathbf{O}\subseteq \mathbf{T}\subseteq \mathbf{A}$.

Given $R\in\mathbf{R}$, we also consider the sets
\[
\mathbf{L}_\diamond(R)=\{L\in\mathbf{L}: L\subseteq R\},\quad \mathbf{L}^\diamond(R)=\{L\in\mathbf{L}: L\supseteq R\},
\]
that is, the set of linear refinements and the set of linear extensions of $R$.
Both those sets may be empty and if $L\in \mathbf{L}$, then $\mathbf{L}_\diamond(L)=\mathbf{L}^\diamond(L)=\{L\}$.

Let $\mathbb{P}(V)$ be the set of the subsets of $V$ and, for $k\in \mathbb{N}$, let $\mathbb{P}_k(V)$ be the set of the subsets of $V$ of size $k$. Given $R\in\mathbf{R}$ and $W\in \mathbb{P}_k(V)$, we say that $W$ is a $k$-{\it maximum set}  for $R$ if, for every $x\in W$ and $y\in V\setminus W$,  we have $x\succeq_R y$.
We define then the set
\[
C_k(R)=\{W\in\mathbb{P}_k(V): W \mbox{ is a $k$-maximum set for }R\}
\]
Clearly, $C_n(R)=\{V\}$ and $C_k(R)=\varnothing$ for all $k>|V|=n$. Thus, the study of the set $C_k(R)$ is meaningful only if $k\in \ldbrack n-1 \rdbrack$.
Note also that if $R$ is reflexive, then $\max(R)=\{x\in V:\{x\}\in C_1(R)\}$, where $\max(R)$ denotes the set of maxima for $R$.
Finally, if $R\in\mathbf{L}$ and $k\in \ldbrack n-1 \rdbrack $, then $|C_k(R)|=1$.

\section{Definition of the flow network method}\label{def-flow}

We call {\it network method} any function from $\mathcal{N}$  to $\mathbf{R}$.
A network method can be seen as a procedure to determine, for every competition and every ordered pair of teams, whether the first team did at least as good as the second one. The {\it flow network method} $\mathfrak{F}$ is defined, for every $N\in\mathcal{N}$,
by
\[
\mathfrak{F}(N)= \{(x,y)\in A: \varphi^N_{xy}\ge \varphi^N_{yx}\}\cup \Delta.
\]
The relation $\mathfrak{F}(N)$ on $V$ is called the {\it flow relation} associated with $N$.
This network method was first proposed by Gvozdik (1987), who focuses only on balanced networks, and later extended and studied in Belkin and Gvozdik (1989).
Note that, for the network $N_{\textsf{C}}$ defined in \eqref{net-comp},  $\mathfrak{F}(N_{\textsf{C}})$ coincides with the relation $\succeq_{\textsf{C}}$ in \eqref{c-intro}.

We call next {\it  network rule} any correspondence from  $\mathcal{N}$ to $\mathbf{L}$. A network rule can be seen as a procedure to determine, for every competition, the set of all the possible rankings of the teams which are consistent with the competition.
The {\it flow network rule} $\mathfrak{F}_\diamond$ is defined, for every $N\in\mathcal{N}$,
by
\[
\mathfrak{F}_{\diamond}(N)=\mathbf{L}_\diamond(\mathfrak{F}(N)).
\]
Such a network rule is the main object studied by Belkin and Gvozdik (1989).

Given now $k\in \ldbrack n-1 \rdbrack $, we finally call {\it $k$-multiwinner network solution} any correspondence from  $\mathcal{N}$ to $\mathbb{P}_k(V)$.
A $k$-multiwinner network solution can be seen as a procedure to determine, for every competition,
 all the possible choices for the set of the best $k$ teams of the competition.
The {\it flow $k$-multiwinner network solution} $\mathfrak{F}_k$ is defined, for every $N\in\mathcal{N}$,
by
\[
\mathfrak{F}_k(N)=C_k(\mathfrak{F}(N)).
\]

Our purpose is to study the properties of the three objects above defined. In order to discuss the properties of the flow network method we just need the crucial fact about complete and quasi-transitive relations which is proved in the next section. For studying the flow network rule and the flow $k$-multiwinner network solution further results from the theory of relations are necessary. They are collected in Section \ref{relation-2}.

Later on, given $N\in\mathcal{N}$ and $x,y\in V$, we write $x\succeq_{N} y$ instead of $x\succeq_{\mathfrak{F}(N)}y$; $x\succ_{N} y$ instead of  $x\succ_{\mathfrak{F}(N)}y$; $x\sim_N y$  instead of $x\sim_{\mathfrak{F}(N)} y$. Thus, $x\succeq_{N} y$ means $x=y$ or $\varphi^N_{x,y}\geq\varphi^N_{yx};$  $x\succ_{N} y$ means $x\neq y$ and $\varphi^N_{xy}>\varphi^N_{yx};$ $x\sim_{N} y$ means $x=y$ or $\varphi^N_{xy}=\varphi^N_{yx}.$

\section{Building complete and quasi-transitive relations}

This section is devoted to the description of a general procedure to construct complete and quasi-transitive relations on $V$.
For every $\vartheta:A\to \mathbb{R}$, define $R(\vartheta)\in \mathbf{R}$ by
\begin{equation*}
R(\vartheta)= \{(x,y)\in A: \vartheta(x,y)\ge \vartheta(y,x)\}\cup \Delta.
\end{equation*}
Given $\vartheta:A\to \mathbb{R}$, we say that $\vartheta$ satisfies the {\it Gomory-Hu condition} if, for every $x,y,z\in V$ distinct, we have
\begin{equation}\label{gomory-eq}
\vartheta(x,z)\ge \min\{\vartheta(x,y),\vartheta(y,z)\}.
\end{equation}
The next proposition is based on Schulze (2011, pp. 277-278).

\begin{proposition}\label{schulze}
If $\vartheta:A\to \mathbb{R}$ satisfies the Gomory-Hu condition, then $R(\vartheta)\in\mathbf{T}$.
\end{proposition}

\begin{proof}
For shortness set $R=R(\vartheta)$. The completeness of $R$ is obvious. We show that $R$ is quasi-transitive.
 Let us consider $x,y,z\in V$ and assume that $x\succ_{R} y$ and  $y\succ_{R} z$.
We must prove that  $x\succ_{R} z$, that is, $x\neq z$ and
\begin{equation}\label{equa3}
\vartheta(x,z)>\vartheta(z,x).
\end{equation}
Note that $x\succ_{R} y$ and  $y\succ_{R} z$  mean $x\neq y$, $y\neq z$,
\begin{equation}\label{equa1}
\vartheta(x,y)>\vartheta(y,x)
\end{equation}
 and
\begin{equation}\label{equa2}
\vartheta(y,z)>\vartheta(z,y).
\end{equation}
We first show that $x\neq z$. Indeed, if it were $x=z$, then we would get $\vartheta(x,y)>\vartheta(y,x)$ and
$\vartheta(y,x)>\vartheta(x,y)$, a contradiction.
By \eqref{gomory-eq}, we have
\begin{equation}\label{gomory-1}
\vartheta(x,z)\ge \min\{\vartheta(x,y),\vartheta(y,z)\}.
\end{equation}
\begin{equation}\label{gomory-2}
\vartheta(z,y)\ge \min\{\vartheta(z,x),\vartheta(x,y)\}.
\end{equation}
\begin{equation}\label{gomory-3}
\vartheta(y,x)\ge \min\{\vartheta(y,z),\vartheta(z,x)\}.
\end{equation}
Assume first that
\begin{equation}\label{eqA}
\vartheta(x,y)\ge \vartheta(y,z).
\end{equation}
Then, \eqref{gomory-1} and \eqref{equa2} give
\begin{equation}\label{eqB}
\vartheta(x,z)
>\vartheta(z,y).
\end{equation}
Moreover, \eqref{eqA} and \eqref{equa2} give $\vartheta(x,y)>\vartheta(z,y).$
Since by \eqref{gomory-2}, $\vartheta(z,y)$ must be greater or equal to one between $\vartheta(z,x)$ and $\vartheta(x,y)$, it follows that $\vartheta(z,y)\ge \vartheta(z,x)$. Using \eqref{eqB} we then get \eqref{equa3}.

Assume next that $\vartheta(x,y)<\vartheta(y,z)$. Then \eqref{gomory-1} and \eqref{equa1} give
\begin{equation}\label{eqC}
\vartheta(x,z)> \vartheta(y,x).
\end{equation}
Moreover, by \eqref{equa1}, $\vartheta(y,x)<\vartheta(y,z)$. Using \eqref{gomory-3}, we get $\vartheta(y,x)\ge \vartheta(z,x)$. Using now \eqref{eqC} we finally get \eqref{equa3}.
\end{proof}

\section{Properties of the flow network method}\label{flow-network-method}

The next result establishes the fundamental property of $\mathfrak{F}$. Such a property was proved by Belkin and Gvozdik (1989, Lemma 2). We derive it as a particular application of Proposition \ref{schulze}. Given $N\in\mathcal{N}$, let $\varphi^N:A\to \mathbb{R}$ be defined, for every $(x,y)\in A$, by $\varphi^N(x,y)=\varphi^N_{xy}$ and note that $\mathfrak{F}(N)=R(\varphi^N)$.

\begin{theorem}\label{F-qt}
Let $N \in\mathcal{N}$. Then $\mathfrak{F}(N)\in\mathbf{T}$.
\end{theorem}

\begin{proof}
 By Proposition \ref{gomory}, $\varphi^N$ satisfies the Gomory-Hu condition. Thus,  by Proposition \ref{schulze}, we get $\mathfrak{F}(N)=R(\varphi^N)\in \mathbf{T}$.
\end{proof}

If $n= 2$, we have $\mathbf{O}=\mathbf{T}$ so that the outcomes of the flow network method are always transitive. That does not hold true when $n\ge 3$, as shown by the next proposition.

\begin{proposition}\label{non-o}
Let $n\ge 3$. Then there exists $N \in\mathcal{N}$  such that $\mathfrak{F}(N)\not \in\mathbf{O}$.
\end{proposition}

\begin{proof}
Let $V=\ldbrack n \rdbrack$. Consider $N=(V,A,c)$ such that $c(1,2)=1$ and $c(i,j)=0$ for all $i,j\in\ldbrack n \rdbrack$ with $i\neq j$ and $(i,j)\neq(1,2)$. Then $\mathfrak{F}(N)$ is such that $1\succ_N 2$ and $i\sim_N j$ for all $i,j\in\ldbrack n \rdbrack$ with $i\neq j$ and $\{i,j\}\neq \{1,2\}$. In particular, we have $2\succeq_N 3$ and $3\succeq_N 1$ but $2\not\succeq_N 1.$ Hence $\mathfrak{F}(N)$ is not transitive and so $\mathfrak{F}(N)\not \in\mathbf{O}$.
\end{proof}

Proposition \ref{main-qt} below states that any relation belonging to $\mathbf{T}$ can in fact be the outcome of the flow network method applied to a suitable network.
Its proof relies on the possibility to naturally associate a network to any relation.
In what follows, given $R\in\mathbf{R}$, we set $N_R=(V, A,c_R)\in\mathcal{N}$ where, for every $(x,y)\in A$, $c_R(x,y)=1$ if $x\succ_{R} y$ and $c_R(x,y)=0$ otherwise.

\begin{lemma}\label{core} Let $R\in \mathbf{T}$ and $(x,y)\in A$. Then:
\begin{itemize}\item[(i)] $\varphi^{N_R}_{xy}>0$ if and only if $x\succ_{R} y;$
\item[(ii)] $\varphi^{N_R}_{xy}=\varphi^{N_R}_{yx}=0$ if and only if $x\sim_R y$.
\end{itemize}
\end{lemma}
\begin{proof} $(i)$ If $x\succ_{R} y,$ then $c_R(x,y)=1$ and in $N_R$ we have at least the path $xy$, so that $\varphi^{N_R}_{xy}\geq 1.$ Assume next that $\varphi^{N_R}_{xy}>0.$ Then there exists a path in $N_R$ from $x$ to $y$. Thus, there exist
 $m\geq 2$ and
 $x_1, \ldots ,x_m\in V$ distinct such that $x_1=x$, $x_m=y$ and,
for every $i\in\ldbrack m-1\rdbrack$, $c_R(x_{i},x_{i+1})=1$. It follows that we have the chain
$x=x_1\succ_{R} x_2\succ_{R} \ldots \succ_{R} x_{m-1}
\succ_{R} x_m=y.$
 Since $S(R)$ is transitive, we also have $x\succ_{R}y.$

$(ii)$ Let $x\sim_R y$. Then we neither have $x\succ_{R} y$ nor $y\succ_{R} x$ and so, by $(i)$, we get $\varphi^{N_R}_{xy}=\varphi^{N_R}_{yx}=0$. Conversely let $\varphi^{N_R}_{xy}=\varphi^{N_R}_{yx}=0$. Then, by $(i)$, we neither have $x\succ_{R} y$ nor $y\succ_{R} x$. By completeness, the only possibility is $x\sim_R y$.
\end{proof}

\begin{proposition}\label{main-qt} If $R\in \mathbf{T}$ then $\mathfrak{F}(N_R)=R.$
\end{proposition}
\begin{proof}  We need to see that, for every  $x,y\in V$ with $x\neq y$,  $x\succeq_{R} y$ if and only if $\varphi^{N_R}_{xy}\geq\varphi^{N_R}_{yx}$. Assume first that $\varphi^{N_R}_{xy}\geq\varphi^{N_R}_{yx}.$ If $\varphi^{N_R}_{xy}>\varphi^{N_R}_{yx},$ then $\varphi^{N_R}_{xy}>0$ and, by Lemma \ref{core}\,(i), we get
$x\succ_{R} y.$ If instead $\varphi^{N_R}_{xy}=\varphi^{N_R}_{yx}$ then, by Lemma \ref{core}\,(i),  we necessarily have $\varphi^{N_R}_{xy}=\varphi^{N_R}_{yx}=0$ and thus, by Lemma \ref{core}\,(ii), $x\sim_R y$ and, in particular, $x\succeq_{R} y$.
 Assume next that  $x\succeq_{R} y$. If $\varphi^{N_R}_{yx}=0,$ we trivially have $\varphi^{N_R}_{xy}\geq\varphi^{N_R}_{yx}.$ But we cannot have $\varphi^{N_R}_{yx}>0$ since otherwise, by Lemma \ref{core}\,(i), we would get $y\succ_{R} x,$ against $x\succeq_{R} y$.
\end{proof}

As a consequence of the above result and Proposition \ref{schulze}, we derive the following interesting
characterization of the complete and quasi-transitive relations.

\begin{theorem}\label{co-qt} Let $R\in\mathbf{R}$. Then $R\in \mathbf{T}$  if and only if there exists $\vartheta:A\to \mathbb{R}$ satisfying the Gomory-Hu condition, such that $R=R(\vartheta).$
\end{theorem}

\begin{proof} One implication is just Proposition \ref{schulze}. Assume now that $R\in \mathbf{T}$. Then, by definition of flow relation, we have $R(\varphi^{N_R})=\mathfrak{F}(N_R)$ and we know, by Proposition \ref{gomory}, that $\varphi^{N_R}$ satisfyies the Gomory-Hu condition. Since, by Proposition \ref{main-qt}, we have $\mathfrak{F}(N_R)=R, $ we deduce that $R=R(\varphi^{N_R}).$
\end{proof}

We now present some further remarkable properties of $\mathfrak{F}$. The first property is about the behaviour of  the relation $\mathfrak{F}(N)$ with respect to a relabelling of its vertices. It simply says that if one decides to  relabel the network vertices, the flow relation accordingly changes. In other words, the flow relation does not depend on the names of vertices. For that reason we call that property neutrality.
The second one is a homogeneity property for the flow relation. We observe that  these two properties are considered by Gonz\'alez-D\'iaz et al. (2014).\footnote{Those authors use the term anonymity instead of neutrality.}

\begin{proposition}\label{neut-F}  Let $N=(V,A,c) \in\mathcal{N}$ and let $\psi:V\rightarrow V$ be a bijection. Then, for every $x,y\in V$, $x\succeq_N y$ if and only if $\psi(x)\succeq_{N^{\psi}} \psi(y).$
\end{proposition}

\begin{proof} It is clearly enough to show that, for every $N=(V,A,c) \in\mathcal{N}$, every $(x,y)\in A$ and every bijection $\psi:V\rightarrow V$, we have
\begin{equation}\label{8}\varphi_{\psi(x)\psi(y)}^{N^{\psi}}=\varphi_{xy}^{N}.
\end{equation}
 Let $S$ be a minimum cut from $\psi(x)$ to $\psi(y)$ in $N^{\psi}$. Then $\psi^{-1}(S)$ is a cut from $x$ to $y$ in $N$. Thus, by definition of $c^{\psi}$ and by \eqref{duality}, we have
$$\varphi_{\psi(x)\psi(y)}^{N^{\psi}}=c^{N^{\psi}}(S)=\sum_{\psi(u)\in S,\, \psi(v)\in S^c} c^{\psi}(\psi(u),\psi(v))=\sum_{u\in \psi^{-1}(S),\, v\in [\psi^{-1}(S)]^c} c(u,v)\geq \varphi_{xy}^{N}.$$
Then, we also have $$\varphi_{xy}^{N}=\varphi_{\psi^{-1}(\psi(x))\,\psi^{-1}(\psi(y))}^{(N^{\psi})^{\psi^{-1}}}\geq \varphi_{\psi(x)\psi(y)}^{N^{\psi}}$$
and so we get the equality \eqref{8}.
\end{proof}

\begin{proposition}\label{homo-net}
Let $N\in \mathcal{N}$ and $\alpha\in \mathbb{N}$. Then  $\mathfrak{F}(\alpha N)=\mathfrak{F}(N)$.
\end{proposition}

\begin{proof} Let $N=(V,A,c)$ and let $\alpha N=(V,A,c')$, where $c'=\alpha c$.
In order to get $\mathfrak{F}(\alpha N)=\mathfrak{F}(N)$  it is enough to show that, for every $(x,y)\in A$, $\varphi_{xy}^{\alpha N}=\alpha\varphi_{xy}^N$.
Let $(x,y)\in A$. Consider $f\in \mathcal{F}(N,x,y)$ such that $\varphi_{xy}^N=\varphi(f)$. Then $\alpha f\in \mathcal{F}(\alpha N,x,y)$ so that
$
\varphi_{xy}^{\alpha N}\ge \varphi(\alpha f)=\alpha  \varphi(f)=\alpha\varphi_{xy}^{N}.
$
Let $S$ be a cut from $x$ to $y$ in $N$ such that $c(S)=\varphi_{xy}^{N}$. Then $S$ is also a cut from $x$ to $y$ in $\alpha N$. Hence, by \eqref{duality},
$
\alpha\varphi_{xy}^{N}=\alpha c(S)=c'(S)\ge \varphi_{xy}^{\alpha N}
$
and the proof is completed.
\end{proof}

Given a network $N$ and two distinct vertices $x^*$ and $y^*$,  it can be useful to have simple conditions which are sufficient to guarantee $x^*\succeq_{N} y^*$ or
$x^*\succ_{N} y^*$. Conditions of that type are described in the next proposition.

\begin{proposition}\label{Eff-F} Let $N=(V,A,c)\in \mathcal{N}$ and $x^*,y^*\in V$ with $x^*\neq y^*$. Assume that $c(x^*,y^*)\ge c(y^*,x^*)$ and that, for every $z\in V\setminus\{x^*,y^*\}$, $c(x^*,z)\ge c(y^*,z)$ and $c(z,y^*)\ge c(z,x^*)$.
Then $x^*\succeq_{N} y^*$.
Assume further that one of the following conditions holds true:
\begin{itemize}
	\item[(a)] $c(x^*,y^*)> c(y^*,x^*)$;
	\item[(b)]for every $z\in V\setminus\{x^*,y^*\}$, $c(x^*,z)> c(y^*,z)$;
	\item[(c)]for every $z\in V\setminus\{x^*,y^*\}$, $c(z,y^*)> c(z,x^*)$;
	\item[(d)] there exists $z\in V\setminus \{x^*,y^*\}$ such that $c(x^*,z)>c(y^*,z)$ and $c(z,y^*)>c(z,x^*)$.
\end{itemize}
Then $x^*\succ_{N} y^*$.
\end{proposition}

\begin{proof}Let $S$ be a minimum cut for $\varphi_{x^*y^*}$ so that $S\subseteq V$, $x^*\in S,\ y^*\notin S$ and $c(S)=\varphi_{x^*y^*}.$ Define $T=(S\setminus \{x^*\})\cup\{y^*\}$ and note that, since $T\subseteq V$, $y^*\in T,\ x^*\notin T,$ we have that $T$ is a cut for $\varphi_{y^*x^*}$. By \eqref{duality}, we have $c(T)\ge \varphi_{y^*x^*}$. In order to obtain $\varphi_{x^*y^*}\ge \varphi_{y^*x^*}$, that is $x^*\succeq_N y^*$, it is enough to show $c(S)- c(T)\ge 0.$
Note that
\[
c(S)=\sum_{u\in S,v\in S^c}c(u,v)=\sum_{u\in S\setminus\{x^*\}, v\in S^c\setminus\{y^*\}}c(u,v)+\sum_{v\in S^c\setminus\{y^*\}}c(x^*,v)+\sum_{u\in S\setminus\{x^*\}}c(u,y^*)+c(x^*,y^*)
\]
and
\[
c(T)=\sum_{u\in T,v\in T^c}c(u,v)=\sum_{u\in T\setminus\{y^*\}, v\in T^c\setminus\{x^*\}}c(u,v)+\sum_{v\in T^c\setminus\{x^*\}}c(y^*,v)+\sum_{u\in T\setminus\{y^*\}}c(u,x^*)+c(y^*,x^*)
\]
Observe now that $T\setminus\{y^*\}=S\setminus\{x^*\}$ and $T^c\setminus\{x^*\}=S^c\setminus\{y^*\}.$
Thus
\begin{equation}\label{diff}
 c(S)-c(T)= (c(x^*,y^*)-c(y^*,x^*))+\sum_{v\in S^c\setminus\{y^*\}}(c(x^*,v)-c(y^*,v))+\sum_{u\in S\setminus\{x^*\}}(c(u,y^*)-c(u,x^*))
\end{equation}
and every term between brackets is non-negative, which says $c(S)- c(T)\ge 0.$

Assume further that one among $(a)$, $(b)$, $(c)$ and $(d)$ holds true. Then, at least one of the terms between brackets in \eqref{diff} is now positive and so $\varphi_{x^*y^*}=c(S)>c(T)\ge\varphi_{y^*x^*}.$
Thus,  $\varphi_{x^*y^*}>\varphi_{y^*x^*}$, that is, $x^*\succ_N y^*$.
\end{proof}

The  flow network method also satisfies a monotonicity property, as described by the next proposition. This property implies that  teams have an incentive to win when the flow network method is used to assess their performance. That is in line with the Property II in Vaziri et al. (2017).

\begin{proposition}\label{Mon1-flow-rule} Let $N=(V,A,c)\in \mathcal{N}$, $N'=(V,A,c')\in \mathcal{N}$ and $x^* \in V$. Assume that the following conditions hold true:
\begin{itemize}
\item[(a)] for every $y\in V\setminus \{x^*\}$, $c'(x^*,y)\geq c(x^*,y)$ and $c'(y,x^*)\le c(y,x^*)$;
\item[(b)] for every $(x,y)\in A$ with $x\neq x^*$ and $y\neq x^*$, $c'(x,y)= c(x,y)$.
\end{itemize}
Then, for every $y\in V\setminus\{x^*\}$ such that $x^*\succeq_{N} y$, we have $x^*\succeq_{N'} y$. Moreover,
for every $y\in V\setminus\{x^*\}$ such that $x^*\succeq_{N} y$ and one among $\ x^*\succ_N y$, $c'(x^*,y)> c(x^*,y)$ and $c'(y,x^*)< c(y,x^*)$ holds true, we have that $x^*\succ_{N'} y$.
\end{proposition}

\begin{proof}Let $y^*\in V\setminus\{x^*\}$ such that $x^*\succeq_{N} y^*$, that is $\varphi_{x^*y^*}^{N}\ge \varphi_{y^*x^*}^{N}$. We want to show that $\varphi_{x^*y^*}^{N'}\ge \varphi_{y^*x^*}^{N'}$.
 We prove first that $\varphi_{x^*y^*}^{N'}\ge \varphi_{x^*y^*}^{N}+c'(x^*,y^*)-c(x^*,y^*).$ Let $f$ be a maximum flow from $x^*$ to $y^*$ in $N$ such that
 \begin{equation}\label{ca}
 \mbox{for every } x\in V\setminus \{x^*\}, \;f(x,x^*)=0.
 \end{equation}
 Such a maximum  flow  can be obtained as the output of the well-known augmenting path algorithm applied to the null flow.
Then we have
\begin{equation}\label{flusso}
\varphi^{N}(f)= \varphi^N_{x^*y^*}= \sum_{y\in V\setminus \{x^*\}}f(x^*,y).
\end{equation}
 Consider now $f':A\to\mathbb{N}_0$ defined by
\[
f'(x,y)=\left\{
\begin{array}{ll}
f(x^*,y^*)+ c'(x^*,y^*)-c(x^*,y^*) &\mbox{if }(x,y)= (x^*,y^*)\\
\vspace{-2mm}\\
f(x,y) &\mbox{if }(x,y)\neq(x^*,y^*)\\
\end{array}
\right.
\]
We show that $f'\in \mathcal{F}(N',x^*,y^*)$. By $(a)$ we have $f'(x,y)\geq 0$ for all $(x,y)\in A$. To see that $f'(x,y)\leq c'(x,y)$ for all $(x,y)\in A$, we need to distinguish  some cases. Surely $f'(x^*,y^*)=f(x^*,y^*)+ c'(x^*,y^*)-c(x^*,y^*)\leq c'(x^*,y^*)$ because $f(x^*,y^*)\leq c(x^*,y^*).$ Let $(x,y)\in A\setminus \{(x^*,y^*)\}$. If $x=x^*$ and $y\neq y^*$ then, by $(a)$, we have
 $f'(x^*,y)=f(x^*,y)\leq c(x^*,y)\leq c'(x^*,y).$
Assume next $x\neq x^*$. If $y=x^*$ then, by \eqref{ca},
$f'(x,y)=f(x,x^*)=0\leq  c'(x,x^*).$
If instead $y\neq x^*$ then, by $(b)$, we have
$f'(x,y)=f(x,y)\leq c(x,y)=c'(x,y).$
We are left to show \eqref{conservation}. Fix $u\not\in \{x^*, y^*\}$. We need to show
\begin{equation}\label{original}
\sum_{z\in V\setminus \{u\}}f'(u,z)=\sum_{z\in V\setminus \{u\}}f'(z,u).
\end{equation}
Since $u\not\in \{x^*, y^*\}$, by $f'$ definition, \eqref{original} is the same as
$
\sum_{(u, z)\in A}f(u,z)=\sum_{(z, u)\in A}f(z,u)
$
which holds true because $f\in\mathcal{F}(N,x^*,y^*)$.
Now, using \eqref{flusso}, we get
\[
\varphi^{N'}(f')=\sum_{y\in V\setminus \{x^*\}}f'(x^*,y)-\sum_{y\in V\setminus \{x^*\}}f'(y,x^*)
\]
\[
=f'(x^*,y^*)+\sum_{y\in V\setminus \{x^*,y^*\}}f'(x^*,y)-\sum_{y\in V\setminus \{x^*,y^*\}}f(y,x^*)
\]
\[
=f(x^*,y^*)+ c'(x^*,y^*)-c(x^*,y^*)+ \sum_{y\in V\setminus \{x^*,y^*\}}f(x^*,y)
=\varphi_{x^*y^*}^N+c'(x^*,y^*)-c(x^*,y^*).
\]
 As a consequence,
\begin{equation}\label{in1}
\varphi_{x^*y^*}^{N'}\ge \varphi^{N'}(f')=\varphi_{x^*y^*}^{N}+c'(x^*,y^*)-c(x^*,y^*).
\end{equation}
Consider now a minimum cut $S$ from $y^*$ to $x^*$ in $N$. Then we have $c(S)=\varphi_{y^*x^*}^{N}.$ Of course, $S$ is also a cut  from $y^*$ to $x^*$ in $N'$ and then
\begin{equation}\label{in3}
c'(S)\ge \varphi_{y^*x^*}^{N'}.
\end{equation}
Moreover we have
\[
c'(S)=\sum_{u\in S,v\in S^c}c'(u,v)=\sum_{u\in S, v\in S^c\setminus \{x^*\}}c'(u,v)+\sum_{u\in S\setminus\{y^*\}}c'(u,x^*)+ c'(y^*,x^*).
\]
Note now that, by the assumptions $(a)$ and $(b)$, we have
\[
\sum_{u\in S, v\in S^c\setminus \{x^*\}}c'(u,v)= \sum_{u\in S, v\in S^c\setminus \{x^*\}}c(u,v),
\qquad
\sum_{u\in S\setminus\{y^*\}}c'(u,x^*)\le
\sum_{u\in S\setminus\{y^*\}}c(u,x^*)
\]
and $c'(y^*,x^*)\le c(y^*,x^*).$
As a consequence,
\begin{equation}\label{in4}
c'(S)+c(y^*,x^*)-c'(y^*,x^*)\le c(S).
\end{equation}
Using now the inequalities  \eqref{in1}, \eqref{in3} and \eqref{in4} and recalling that $\varphi^N_{x^*y^*}\ge \varphi^N_{y^*x^*}$, we obtain
\[
\varphi_{x^*y^*}^{N'}\ge \varphi_{x^*y^*}^{N}+c'(x^*,y^*)-c(x^*,y^*)\ge \varphi_{y^*x^*}^{N}+c'(x^*,y^*)-c(x^*,y^*)
\]
\[
=c(S)+c'(x^*,y^*)-c(x^*,y^*)\ge  c'(S)+\left(c(y^*,x^*)-c'(y^*,x^*)\right)+\left(c'(x^*,y^*)-c(x^*,y^*)\right)
\]
\[
\ge \varphi_{y^*x^*}^{N'}+\left(c(y^*,x^*)-c'(y^*,x^*)\right)+\left(c'(x^*,y^*)-c(x^*,y^*)\right),
\]
and thus
\[
\varphi_{x^*y^*}^{N'}
\ge \varphi_{y^*x^*}^{N'}+\left(c(y^*,x^*)-c'(y^*,x^*)\right)+\left(c'(x^*,y^*)-c(x^*,y^*)\right).
\]
In particular, $\varphi_{x^*y^*}^{N'}
\ge \varphi_{y^*x^*}^{N'}$ so that the proof of the first part of the theorem is complete.

Consider now $y^*\in V\setminus\{x^*\}$ such that $x^*\succeq_{N} y^*$ and
one among $x^*\succ_N y^*$, $c'(x^*,y^*)> c(x^*,y^*)$ and $c'(y^*,x^*)< c(y^*,x^*)$ holds true.
The same argument as above gives  $\varphi_{x^*y^*}^{N'}
> \varphi_{y^*x^*}^{N'}$ which completes the proof.
\end{proof}

We now observe that reversing a network, the flow relation gets its reversal. We call such a property reversal symmetry due to its similarity
to the concept of reversal symmetry for social welfare functions proposed by Saari (1994) and recently deepened by Bubboloni and Gori (2015). Note also that this property is called inversion by Gonz\'alez-D\'iaz et al. (2014). The proof of this fact, formalized in the proposition below, is very easy and thus omitted.

\begin{proposition}\label{Reversal}
Let $N \in\mathcal{N}$. Then, for every $(x,y)\in A$, we have $\varphi_{yx}^{N^r}=\varphi_{xy}^{N}$.
In particular, $\mathfrak{F}(N^r)=\mathfrak{F}(N)^r$.
\end{proposition}

The next proposition shows that if in a competition the number of wins (losses) of a team is greater than the number of its losses (wins), then the flow network method does not make that team be the worst (best) one. Its proof immediately follows from Theorem 1 in Lov\'asz (1973) once adapted to the network context\footnote{For a more general approach, see also the main result in Hartmann and Schneider (1993).}, and from Proposition \ref{Reversal} after having observed that $o^{N^r}(x)=i^{N}(x)$ for all $x\in V$.

\begin{proposition}\label{szigeti}
Let $N \in\mathcal{N}$ and $x^*\in V$. If $o(x^*)>i(x^*)$,  then there exists $y\in V\setminus \{x^*\}$ such that $x^*\succ_{N}y$.  If $o(x^*)<i(x^*)$,  then there exists $y\in V\setminus \{x^*\}$ such that $y\succ_{N}x^*$.
\end{proposition}

A network method is said {\it flat} on a network $N$ if it associates with $N$ the total relation $V^2$.
A network method is said {\it symmetric} if it is flat on every network $N$ such that $N=N^r$. The symmetry property is a basic requirement for network methods invoked by Gonz\'alez-D\'iaz et al. (2014). The next proposition fully characterizes the set of networks on which $\mathfrak{F}$ is flat. Generalizing the classic definition given for directed graphs, we define $N\in \mathcal{N}$ {\it pseudo-symmetric} if, for every $x\in V$, $o(x)=i(x)$.

\begin{proposition}\label{flat}Let  $N \in\mathcal{N}$.
$\mathfrak{F}$ is  flat on $N$ if and only if $N$ is pseudo-symmetric. In particular, $\mathfrak{F}$ is a symmetric network method.\footnote{We thank L\'aszl\'o Lov\'asz for useful advices about this proposition.}
\end{proposition}

\begin{proof} Let $N \in\mathcal{N}$ be pseudo-symmetric. We want to show that, for every $x,y\in V$ with $x\neq y$, we have $\varphi_{yx}^{N}=\varphi_{xy}^{N}.$
To that purpose we first show that, for every proper nonempty subset $S$ of $V$, we have $c(S)=c(S^c)$\footnote{That fact seems to be known in the literature even though we could not find a precise reference. Thus, for completeness, we provide a proof.}. Let $S\subseteq V$ with $\varnothing\neq S\neq V$.  Given $x\in S$, we have that
\[
o(x)=\sum_{y\in V\setminus\{x\}}c(x,y)=\sum_{y\in S^c}c(x,y)+\sum_{y\in S\setminus\{x\}}c(x,y)
\]
while
\[
i(x)=\sum_{y\in V\setminus\{x\}}c(y,x)=\sum_{y\in S^c}c(y,x)+\sum_{y\in S\setminus\{x\}}c(y,x).
\]
Thus, we have
\[
\sum_{x\in  S}o(x)=\sum_{x\in  S}\sum_{y\in S^c}c(x,y)+\sum_{x\in  S}\sum_{y\in S\setminus\{x\}}c(x,y)=c(S)+\sum_{(x,y)\in A\cap (S\times S)}c(x,y)
\]
and
\[
\sum_{x\in  S}i(x)=\sum_{x\in  S}\sum_{y\in S^c}c(y,x)+\sum_{x\in  S}\sum_{y\in S\setminus\{x\}}c(y,x)=c(S^c)+\sum_{(x,y)\in A\cap (S\times S)}c(y,x)
\]
\[
=c(S^c)+\sum_{(x,y)\in A\cap (S\times S)}c(x,y).
\]
Since, for every $x\in S$, we have $o(x)=i(x)$, it follows that
\[0=\sum_{x\in  S}(o(x)-i(x))=c(S)-c(S^c).
\]
Let now $x,y\in V$ with $x\neq y$. The map from $\mathcal{C}(N,x,y)$ to $\mathcal{C}(N,y,x)$ associating with $S\in \mathcal{C}(N,x,y)$ the set $S^c\in \mathcal{C}(N,y,x)$ is a bijection. It follows that
\[\varphi_{xy}=\min_{S\in \mathcal{C}(N,x,y)} c(S)=\min_{S\in \mathcal{C}(N,x,y)} c(S^c)=\min_{T\in \mathcal{C}(N,y,x)} c(T)=\varphi_{yx}.
\]

Let next $\mathfrak{F}$ be  flat on $N$. Assume, by contradiction, that there exists $x^*\in V$ with $o(x^*)\neq i(x^*).$ Then, by Proposition \ref{szigeti}, there exists $y\in V$ such that $y\not \sim_N x^*$ against $\mathfrak{F}(N)=V^2.$

Finally, to show that $\mathfrak{F}$ is symmetric, note that if $N\in \mathcal{N}$ is such that  $N=N^r$ then, for every $x\in V$, $o^N(x)=i^{N^r}(x)=i^N(x)$ and thus $N$ is pseudo-symmetric.
\end{proof}

Given two networks $N=(V,A,c)$ and $N'=(V,A,c')$, we define $N+N'=(V,A,c+c')\in\mathcal{N}$. From Proposition \ref{flat}, we immediately obtain for the flow network method the flatness preservation property introduced by Gonz\'alez-D\'iaz et al. (2014).

\begin{proposition}
Let $N,N'\in\mathcal{N}$. If $\mathfrak{F}$ is flat on $N$ and $N'$, then $\mathfrak{F}$ is flat on $N+N'$.
\end{proposition}

\subsection{The flow, the Borda and the dual Borda network methods}

Let us consider now the Borda network method $\mathfrak{B}:\mathcal{N}\to\mathbf{R}$ and the dual Borda network method $\widehat{\mathfrak{B}}:\mathcal{N}\to\mathbf{R}$ respectively defined, for every $N\in\mathcal{N}$, by
$\mathfrak{B}(N)=R(o^N)$ and $\widehat{\mathfrak{B}}(N)=R(-i^N).$ In other words, $\mathfrak{B}(N)$ and $\widehat{\mathfrak{B}}(N)$ are the relations respectively induced by the outdegree function and by the opposite of the indegree function associated with the network $N$. Note that $\mathfrak{B}(N),  \widehat{\mathfrak{B}}(N)\in \mathbf{O}$ and that $\widehat{\mathfrak{B}}(N)=(R(i^N))^r$.
Our purpose is to describe some links among the flow, the Borda and the dual Borda network methods.

We begin with the next lemma about balanced networks. It was stated, without proof, by Gvozdik (1987, Proposition 1) and Belkin and Gvozdik (1989, Lemma 1).

\begin{proposition}\label{BeGv-lemma1}
Let $N=(V,A,c)\in \mathcal{B}$. Then, for every $x,y\in V^2$ with $x\neq y$, $\varphi_{xy}-\varphi_{yx}=o(x)-o(y)$.
\end{proposition}

\begin{proof} It is enough to show that, for every $x,y\in V$ with $x\neq y$, $\varphi_{xy}-\varphi_{yx}\ge o(x)- o(y)$.
Let  $x,y\in V$ with $x\neq y$. Let $S$ be a minimum cut for $\varphi_{xy}$, so that $S\subseteq V$, $x\in S,\ y\notin S$ and $c(S)=\varphi_{xy}$. Define $T=(S\setminus \{x\})\cup\{y\}$ and note that, since $T\subseteq V$, $y\in T,\ x\notin T,$ we have that $T$ is a cut for $\varphi_{yx}$. By \eqref{duality}, we get $c(T)\ge \varphi_{yx}$. As a consequence $\varphi_{xy}-\varphi_{yx}\ge c(S)-c(T)$.
Note that
\[
c(S)=\sum_{u\in S,v\in S^c}c(u,v)=\sum_{u\in S\setminus\{x\}, v\in S^c\setminus\{y\}}c(u,v)+\sum_{v\in S^c}c(x,v)+\sum_{u\in S\setminus\{x\}}c(u,y)
\]
and
\[
c(T)=\sum_{u\in T,v\in T^c}c(u,v)=\sum_{u\in S\setminus\{x\}, v\in S^c\setminus\{y\}}c(u,v)+\sum_{v\in (S^c\setminus\{y\})\cup\{x\}}c(y,v)+\sum_{u\in S\setminus\{x\}}c(u, x).
\]
Let $k\in \mathbb{N}_0$ be the balance of $N.$ Then we have $c(u,x)=k-c(x,u)$ for all $u\in V\setminus \{x\}$ and $c(u,y)=k-c(y,u)$ for all $u\in V\setminus \{y\}.$ It follows that
\[
\varphi_{xy}-\varphi_{yx}\ge c(S)-c(T)=\sum_{v\in S^c}c(x,v) +\sum_{u\in S\setminus\{x\}}c(u,y) -\sum_{v\in (S^c\setminus\{y\})\cup\{x\}}c(y,v)-\sum_{u\in S\setminus\{x\}}c(u, x) =
\]
\[
\sum_{v\in S^c}c(x,v) +\sum_{u\in S\setminus\{x\}}k-\sum_{u\in S\setminus\{x\}}c(y,u) -\sum_{v\in (S^c\setminus\{y\})\cup\{x\}}c(y,v)-\sum_{u\in S\setminus\{x\}}k+\sum_{u\in S\setminus\{x\}}c( x,u)=
\]
\[
\sum_{u\in V\setminus\{x\}}c(x,u)-\sum_{u\in V\setminus\{x\}}c(y,u)=o(x)-o(y).
\]
\end{proof}

As already pointed out by Belkin and Gvozdik (1989), from Proposition \ref{BeGv-lemma1} we can immediately deduce that the flow and the Borda network methods agree on balanced networks. Indeed, we can also prove that on those networks they both agree with the dual Borda network method too.

\begin{proposition}\label{balan-ut}
Let $N\in\mathcal{B}$. Then  $\mathfrak{F}(N)=\mathfrak{B}(N)=\widehat{\mathfrak{B}}(N)$.
\end{proposition}

\begin{proof} Let $N=(V,A,c)$. By Proposition \ref{BeGv-lemma1}, we immediately have that $\mathfrak{F}(N)=R(o)=\mathfrak{B}(N)$. On the other hand, if  $k\in\mathbb{N}_0$ is the balance of $N$, we have $o(x)+i(x)=k(n-1)$ for all $x\in V$. Thus, $\mathfrak{B}(N)=R(o)=R(-i)=\widehat{\mathfrak{B}}(N)$.
\end{proof}

Since for $n=2$, we have $\mathcal{N}=\mathcal{B}$, we immediately get the next corollary.

\begin{corollary}\label{cor-n2}
Let $n=2$. Then, for every $N\in\mathcal{N}$, $\mathfrak{F}(N)=\mathfrak{B}(N)=\widehat{\mathfrak{B}}(N)$.
\end{corollary}

The following Propositions \ref{suff-cond-FB}, \ref{suff-cond-FBhat}, \ref{cor-new} and Theorem \ref{OI} deepen the comparison among the flow, the Borda and the dual Borda network methods. In particular, they show that those methods agree on a set of networks larger than the one of the balanced networks. Of course there are instead cases in which the Borda and the dual Borda network methods do not coincide with the flow network method. This happens for every  $N\in \mathcal{N}$ for which $\mathfrak{F}(N)\notin \mathbf{O}$.

\begin{proposition}\label{suff-cond-FB}
Let $n\ge 3$ and $N=(V,A,c)\in \mathcal{N}$. Assume that, for every distinct $x,y,u\in V$, there exists a real number $\beta(x,y,u)> -\frac{1}{n-2}$ such that
\[
c(y,u)-c(x,u)+c(u,y)-c(u,x)\ge \beta(x,y,u) (o(x)-o(y)).
\]
Then $\mathfrak{F}(N)=\mathfrak{B}(N)$.
\end{proposition}

\begin{proof}
In order to show that $\mathfrak{F}(N)=\mathfrak{B}(N),$ it is enough to show that, for every $x,y\in V$ with $x\neq y$, $o(x)\ge o(y)$ implies $\varphi_{xy}\ge \varphi_{yx}$ and $o(x)> o(y)$ implies $\varphi_{xy}> \varphi_{yx}$. For that purpose we will show that, for every $x,y\in V$ with $x\neq y$, there exists a positive constant $\alpha(x,y)$ such that
$\varphi_{xy}-\varphi_{yx}\geq \alpha(x,y) (o(x)-o(y))$.

Fix $x,y\in V$ with $x\neq y$ and let $S$ be a minimum cut from $x$ to $y$ so that $c(S)=\varphi_{xy}$.
Consider now $T=(S\setminus \{x\})\cup\{y\}$ and note that, since $T\subseteq V$, $y\in T$ and $ x\notin T$, we have that $T$ is a cut from $y$ to $x$. Then we have that
\[
c(S)=\sum_{u\in S,v\in S^c}c(u,v)=\sum_{u\in S\setminus\{x\}, v\in S^c\setminus\{y\}}c(u,v)+\sum_{v\in S^c}c(x,v)+\sum_{u\in S\setminus\{x\}}c(u,y)
\]
and
\[
c(T)=\sum_{u\in T,v\in T^c}c(u,v)=\sum_{u\in S\setminus\{x\}, v\in S^c\setminus\{y\}}c(u,v)+\sum_{v\in (S^c\setminus\{y\})\cup\{x\}}c(y,v)+
\sum_{u\in S\setminus\{x\}}c(u, x).
\]
Since $c(T)\ge \varphi_{yx}$, we then have
\[
\varphi_{xy}-\varphi_{yx}\ge
c(S)-c(T)=\sum_{v\in S^c}c(x,v) +\sum_{u\in S\setminus\{x\}}c(u,y) -\sum_{v\in (S^c\setminus\{y\})\cup\{x\}}c(y,v)-\sum_{u\in S\setminus\{x\}}c(u, x)
\]
\[
=o(x)-o(y)+\sum_{u\in S\setminus\{x\}}\left(c(y,u)-c(x,u)+c(u,y)-c(u,x)\right)
\]
\[
\ge o(x)-o(y)+ \sum_{u\in S\setminus\{x\}} \beta(x,y,u)(o(x)-o(y))=\left(1+\sum_{u\in S\setminus\{x\}} \beta(x,y,u)\right)(o(x)-o(y)).
\]
Thus, we are left with proving that
$
1+\sum_{u\in S\setminus\{x\}} \beta(x,y,u)>0.
$
Indeed, since $|S|\le n-1$, we have
\[
1+\sum_{u\in S\setminus\{x\}} \beta(x,y,u)> 1-\frac{|S|-1}{n-2}\ge 0,
\]
and the proof is completed.
\end{proof}

\begin{proposition}\label{suff-cond-FBhat}
Let $n\ge 3$ and $N=(V,A,c)\in \mathcal{N}$. Assume that, for every distinct $x,y,u\in V$, there exists a real number $\gamma(x,y,u)> -\frac{1}{n-2}$ such that
\[
c(x,u)-c(y,u)+c(u,x)-c(u,y)\ge \gamma(x,y,u) (i(y)-i(x)).
\]
Then $\mathfrak{F}(N)=\widehat{\mathfrak{B}}(N)$.
\end{proposition}
The proof of the above result is similar to the one of Proposition \ref{suff-cond-FB} and thus omitted.

\begin{proposition}\label{cor-new}

Let $n\ge 3$ and $N=(V,A,c)\in \mathcal{N}$. Assume that there exist $a,b,l\in\mathbb{N}_0$ with $a+b>0$ and $\omega:V\to \mathbb{N}_0$  such that, for every $(x,y)\in A$, we have
$
c(x,y)=a \omega(x)+b \omega(y)+l.
$
If $(n-1)a<b$ or $(n-1)b<a$,
then $\mathfrak{F}(N)=\mathfrak{B}(N)=\widehat{\mathfrak{B}}(N)$.
\end{proposition}

\begin{proof}Let us prove first that $\mathfrak{F}(N)=\mathfrak{B}(N)$.  Let $x,y,u\in V$ be distinct and show that there exists $\beta=\beta(x,y,u)$ such that $\beta> -\frac{1}{n-2}$ and
\[
c(y,u)-c(x,u)+c(u,y)-c(u,x)\ge \beta (o(x)-o(y)),
\]
in order to apply Proposition \ref{suff-cond-FB}.
Note that
$
c(y,u)-c(x,u)+c(u,y)-c(u,x)=(a+b)(\omega(y)-\omega(x))
$
and that
$
o(x)-o(y)=(b-(n-1)a)(\omega(y)-\omega(x)).
$
Then we have to show that there exists $\beta> -\frac{1}{n-2}$ such that
\[
(a+b)(\omega(y)-\omega(x))\ge \beta (b-(n-1)a)(\omega(y)-\omega(x)).
\]
If $\omega(y)-\omega(x)\ge 0$, we can choose $\beta=0$. If $\omega(y)-\omega(x)<0$, then it has to be
$
a+b\le \beta (b-(n-1)a).
$
We cannot have $b-(n-1)a=0$.  Namely if $(n-1)a<b$, then $b-(n-1)a>0$. Let instead $(n-1)b<a$ and assume, by contradiction, that $b-(n-1)a=0$. Then we have $(n-1)^2a<a$. Thus $a\neq 0$ and so, since $a\in \mathbb{N}_0,$ we have $a>0$. It follows that $n(n-2)<0$ against $n\geq 3$.
Then it is meaningful  to define
\[
\beta = \frac{a+b}{b-(n-1)a}.
\]
If $(n-1)a<b$, then we  have $\beta >0>-\frac{1}{n-2}.$
If instead $(n-1)b<a$, then this implies $\beta> -\frac{1}{n-2}.$

Next we  show that $\mathfrak{F}(N)=\widehat{\mathfrak{B}}(N)$.
Let $x,y,u\in V$ be distinct and  let us show that there exists $\gamma=\gamma(x,y,u)$ such that $\gamma> -\frac{1}{n-2}$ and
\[
c(x,u)-c(y,u)+c(u,x)-c(u,y)\ge \gamma (i(y)-i(x)),
\]
in order to apply Proposition \ref{suff-cond-FBhat}.
Note that
$
c(x,u)-c(y,u)+c(u,x)-c(u,y)=(a+b)(\omega(x)-\omega(y))
$
and that
$
i(y)-i(x)=(a-(n-1)b)(\omega(x)-\omega(y)).
$
Then we have to show that there exists $\gamma> -\frac{1}{n-2}$ such that
\[
(a+b)(\omega(x)-\omega(y))\ge \gamma (a-(n-1)b)(\omega(x)-\omega(y)).
\]
If $\omega(x)-\omega(y)\ge 0$, we can choose $\gamma=0$. If $\omega(x)-\omega(y)<0$, then it has to be
$
(a+b)\le \gamma (a-(n-1)b).
$
As before, it cannot be $a-(n-1)b=0$.
Thus  it is meaningful  to define
\[
\gamma = \frac{a+b}{a-(n-1)b}.
\]
If $(n-1)b<a$, then we have $\gamma = \frac{a+b}{a-(n-1)b}>0>-\frac{1}{n-2}.$ If $(n-1)a<b$, then this implies $\gamma >-\frac{1}{n-2}.$
\end{proof}

Let us finally define two special sets of networks, namely
\[
\mathcal{O}=\{(V,A,c)\in \mathcal{N}: \exists\mbox{ $\omega_1:V\to \mathbb{N}_0$ such that, $\forall\;(x,y)\in A$, $c(x,y)=\omega_1(x)$}\},
\]
\[
\mathcal{I}=\{(V,A,c)\in \mathcal{N}: \exists\mbox{ $\omega_2:V\to \mathbb{N}_0$ such that, $\forall\;(x,y)\in A$, $c(x,y)=\omega_2(y)$}\}.
\]

\begin{theorem}\label{OI}
Let $N\in \mathcal{B}\cup\mathcal{O}\cup \mathcal{I}$.
Then $\mathfrak{F}(N)=\mathfrak{B}(N)=\widehat{\mathfrak{B}}(N)$.
\end{theorem}

\begin{proof}If  $n=2$, then apply Corollary \ref{cor-n2}. Assume now that $n\ge 3$. If $N\in \mathcal{B}$, then apply Proposition \ref{balan-ut}.
 If $N\in \mathcal{O}\cup \mathcal{I}$, then apply Proposition
 \ref{cor-new} with $a=1$, $b=0$, $l=0$ and with $a=0$, $b=1$, $l=0$.
\end{proof}
Consider now  the set of networks $\mathcal{C}=\{(V,A,c)\in \mathcal{N}: c \mbox{ is constant}\}$.
It is easily proved that if  $n=2$, then $\mathcal{C}\subsetneq\mathcal{B}=\mathcal{O}= \mathcal{I}= \mathcal{N};$  if $n\geq 3$, then
$\mathcal{O}\cap \mathcal{I}=\mathcal{O}\cap \mathcal{B}=\mathcal{I}\cap \mathcal{B}=\mathcal{C}.$ In particular,  $\mathcal{O}$ and $\mathcal{I}$ have a very small intersection with the set of balanced networks, so that $ \mathcal{B}\cup\mathcal{O}\cup \mathcal{I}$ is, in general, considerably larger than $ \mathcal{B}.$

\section{Further properties of relations}\label{relation-2}

In this section we collect several facts about relations which are fundamental for the analysis of the flow network rule and the flow $k$-multiwinner network solution.
We stress that all the theorems and propositions of the section do not require that $V$ has at least two elements, but only need that $V$ is finite and nonempty.

Let us begin with recalling a classic result about acyclic relations which will turn out to be crucial in the sequel (Szpilrajn, 1939).

\begin{theorem}\label{greco}
Let $R\in \mathbf{R}$. Then the following facts hold:
\begin{itemize}
\item[(a)]$R$ is acyclic if and only if $\mathbf{L}^\diamond(R)\neq \varnothing$.
\item[(b)]If $R$ is acyclic and $x,y\in V$ are such that $(x,y)\not\in R$ and $(y,x)\not\in R$, then there exists $L\in\mathbf{L}^\diamond(R)$ such that $x\succ_{L} y$.
\end{itemize}
\end{theorem}

The proofs of the next three useful facts are almost immediate. We give the proof of the third one only.

\begin{proposition} \label{3-bis}
Let $R\in\mathbf{R}$,  $k\in \ldbrack n-1 \rdbrack $ and $x^*\in V$. Assume that there exist distinct $y_1,\ldots, y_{n-k}\in V\setminus \{x^*\}$ such that,
for every $j\in\ldbrack n-k\rdbrack$, $x^*\succ_R y_j$. Then, for every $W\in C_k(R)$,  $x^*\in W$.
\end{proposition}

\begin{proposition} \label{ckrev}
Let $R\in\mathbf{R}$ and $k\in  \ldbrack n-1 \rdbrack $. Then $W\in  C_k(R^r)$ if and only if $V\setminus W\in C_{n-k}(R).$
\end{proposition}

\begin{proposition}\label{totale} Let
$R\in\mathbf{R}$ be reflexive and $k\in  \ldbrack n-1 \rdbrack $. Then  $C_k(R)=\mathbb{P}_k(V)$ if and only if $R=V^2.$
\end{proposition}
\begin{proof} Let $R=V^2$ and pick $W\in \mathbb{P}_k(V)$. Then, for every $x\in W$ and every $y\in V\setminus W$, we have $x\succeq_Ry$ and thus $W\in C_k(R).$ Conversely let $C_k(R)=\mathbb{P}_k(V)$. Assume, by contradiction, that there exist $x\in V$ and $y\in V$ such that $x\not \succeq_Ry$. Since $R$ is reflexive, we necessarily have $x\neq y.$ Since $k\in  \ldbrack n-1 \rdbrack $,  there  exists $W\in \mathbb{P}_k(V)$ with $x\in W$ and $y \in V\setminus W$. The fact that $x\not \succeq_Ry$ says that $W\notin C_k(R)$ against $C_k(R)=\mathbb{P}_k(V)$.
\end{proof}

We propose now several results about complete and quasi-acyclic relations. We emphasize that such results, establishing the general properties of the set
$\mathbf{A}$, provide tools to control the properties of any possible network method with values in $\mathbf{A}$.

\begin{proposition} \label{lin-ref}
Let $R\in\mathbf{A}$. Then $\mathbf{L}_\diamond(R)= \mathbf{L}^\diamond(S(R))\neq\varnothing$.
\end{proposition}

\begin{proof} We first show that $\mathbf{L}^\diamond(S(R))\subseteq \mathbf{L}_\diamond(R).$ Let $L\in \mathbf{L}^\diamond(S(R)).$ Thus $L$ is linear and $L\supseteq S(R).$ We want to show that $L\subseteq R.$
 Pick $(x,y)\in L$. If $x=y$, then $(x,x)\in R$ because $R$ is complete and then, in particular, reflexive. Let then $x\neq y$.
Assume, by contradiction, that $(x,y)\not\in R$.  Since $R$ is complete, we have $(y,x)\in R$ and then $(y,x)\in S(R)$. Thus $(y,x)\in L$ and, since $L$ is antisymmetric, we get the contradiction $x=y$.

We next show that $\mathbf{L}_\diamond(R)\subseteq \mathbf{L}^\diamond(S(R)).$ Let $L\in \mathbf{L}_\diamond(R)$. Thus $L$ is linear and $L\subseteq R.$ We need to show that $S(R)\subseteq L.$ Let $(x,y)\in S(R)$, so that we have  $(x,y)\in R$, $(y,x)\notin R$. Assume, by contradiction, that $(x,y)\notin L$. Then, since $L$ is complete, we have $(y,x)\in L$ and so also $(y,x)\in R,$ a contradiction.
Finally, since $R\in\mathbf{A}$, we have that $S(R)$ is acyclic so that, by Theorem \ref{greco},  $\mathbf{L}^\diamond(S(R))\neq\varnothing$.
\end{proof}

\begin{proposition} \label{2-bis}
Let $R\in\mathbf{A}$, $x^*\in V$ and $W\subseteq V\setminus \{x^*\}$. Assume that, for every $y\in W$, $x^*\succ_R y$. Then, for every
 $L\in \mathbf{L}_{\diamond}(R)$ and $y\in W$, $x^*\succ_L y$.
\end{proposition}

\begin{proof}  Let $L\in \mathbf{L}_{\diamond}(R).$ Then, by Proposition \ref{lin-ref}, we have that  $L\in \mathbf{L}^\diamond(S(R))$, that is, $L\supseteq S(R).$ Thus,  $x^*\succ_R y$ for all $y\in W$ implies $x^*\succeq_L y$ for all $y\in W$. On the other hand, since $W\subseteq V\setminus \{x^*\}$, we necessarily have $x^*\succ_L y$ for all $y\in W$.
\end{proof}

\begin{proposition} \label{2}
Let $R\in\mathbf{A}$, $x^*\in V$ and $W\subseteq V\setminus \{x^*\}$. Assume that, for every $y\in W$, $x^*\succeq_R y$. Then there exists $L\in \mathbf{L}_{\diamond}(R)$ such that, for every $y\in W$, $x^*\succ_L y$.
\end{proposition}

\begin{proof}
Let $Z=\{z\in W: z\sim_R x^*\}$ and note that $x^*\notin Z.$
Then $W\setminus Z=\{y\in W: x^*\succ_R y\}\subseteq V\setminus \{x^*\}$ and, by Proposition \ref{2-bis}, for every $L\in \mathbf{L}_\diamond(R)$ and $y\in W\setminus Z$, we have $x^*\succ_L y$.
Thus, if $Z=\varnothing$, it is sufficient to pick any $L\in \mathbf{L}_\diamond(R)$.
Assume next that $Z\neq \varnothing.$ By the above observation, we only need to find $L\in \mathbf{L}_\diamond(R)$  such that for every $z\in Z$, $x^*\succ_L z$. Let $|Z|=m$ for some $1\leq m\leq |W|$ and let $Z=\{z_1,\dots, z_m\}$. Define, for every $i\in\{2,\dots, m+1\}$, $U_i(R)=\{(x^*,z_j):  j\in \ldbrack i-1\rdbrack\}$, and consider the $m+1$ relations on $V$ given by $R_1=S(R)$ and $R_i=S(R)\cup U_i(R)$.
Note that, for every $i\in \ldbrack m \rdbrack$, $(x^*,z_i)\not\in R_i$ and $(z_i,x^*)\not\in R_i$.  Indeed, since $z_i\in Z$, we have that $(x^*,z_i)\notin S(R)$ and $(z_i,x^*)\notin S(R)$ and, by $U_i(R)$ definition, we also  have $(x^*,z_i)\notin U_i(R), (z_i,x^*)\notin U_i(R).$

Let $I= \ldbrack m+1 \rdbrack$. We claim that, for every $i\in I$, $ \mathbf{L}_\diamond(R)\cap  \mathbf{L}^\diamond(R_i)\neq \varnothing.$ Assume, by contradiction, that there exists $i\in I$ such that $ \mathbf{L}_\diamond(R)\cap  \mathbf{L}^\diamond(R_i)=\varnothing.$ Let
$s=\min\{i\in I:  \mathbf{L}_\diamond(R)\cap  \mathbf{L}^\diamond(R_i)=\varnothing\}$
and note that, since by Proposition \ref{lin-ref}, $ \mathbf{L}_\diamond(R)\cap  \mathbf{L}^\diamond(R_1)= \mathbf{L}_\diamond(R)\cap  \mathbf{L}^\diamond(S(R))=\mathbf{L}_\diamond(R)\neq\varnothing,$ we surely have $s\geq 2$. By definition of $s$, we then have $\mathbf{L}_\diamond(R)\cap  \mathbf{L}^\diamond(R_{s-1})\neq \varnothing.$ From  $ \mathbf{L}^\diamond(R_{s-1})\neq \varnothing$, by Theorem \ref{greco}\,(a), we deduce that $R_{s-1}$ is acyclic. Moreover $x^*$ and $z_{s-1}$ are distinct $(x^*,z_{s-1})\not\in R_{s-1}$ and  $(z_{s-1},x^*)\not\in R_{s-1}$. Thus, by Theorem \ref{greco}\,(b), there exists  $L\in  \mathbf{L}^\diamond(R_{s-1})$ such that $x^*\succ_L  z_{s-1}.$ Since
$R_{s}=R_{s-1}\cup \{(x^*,z_{s-1})\}$, this says that $L\in  \mathbf{L}^\diamond(R_{s})$. From $R_s\supseteq S(R)$ and from Proposition \ref{lin-ref}, we then deduce that $L\in \mathbf{L}^\diamond(S(R))=\mathbf{L}_\diamond(R)$. Hence $L\in \mathbf{L}_\diamond(R)\cap  \mathbf{L}^\diamond(R_s)$, against the definition of $s.$

Now pick $L\in  \mathbf{L}_\diamond(R)\cap  \mathbf{L}^\diamond(R_{m+1}).$ Then we have $L\in  \mathbf{L}_\diamond(R)$ and  $x^*\succ_L z_i$ for all $i\in  \ldbrack m\rdbrack$, as required.
\end{proof}

Let $W$ be a nonempty finite subset of $V$. Recall that the restriction of $R$ to $W$ is the relation on $W$ defined by $R_{\mid W}=R\cap W^2$.

\begin{proposition} \label{ckR}
Let $R\in\mathbf{A}$ and $k\in \ldbrack n-1 \rdbrack $. Then
\[
C_k(R)=\bigcup_{L\in\mathbf{L}_\diamond(R)}C_k(L).
\]
In particular, $1\leq |C_k(R)|\leq |\mathbf{L}_\diamond(R)|$ and $C_k(R)\neq\varnothing$.
\end{proposition}

\begin{proof}Let $C=\bigcup_{L\in\mathbf{L}_\diamond(R)}C_k(L)$.
By Proposition \ref{lin-ref}, we have that $\mathbf{L}_\diamond(R)\neq \varnothing$
and, for every $L\in\mathbf{L}_\diamond(R)$, we have
$|C_k(L)|=1$. Thus surely $1\leq|C|\leq |\mathbf{L}_\diamond(R)|.$ We next show that $C_k(R)=C$.
Let $W\in C$. Then there exists $L\in \mathbf{L}_\diamond(R)$ such that $W\in C_k(L)$. Thus, for every $x\in W$ and $y\in V\setminus W$, we have $x\succeq_L y$ and so, since $L$ is a refinement of $R$, we also have  $x\succeq_R y$. Hence $W\in C_k(R).$
Consider now $W\in C_k(R)$. Since $R_{\mid W}$ is a complete and acyclic relation on $W$ and $R_{\mid V\setminus W}$ is a complete and acyclic relation on $V\setminus W$, by Proposition \ref{lin-ref}, we have that both admit linear refinements. Let $L_1$ be a linear refinement of  $R_{\mid W}$ and $L_2$ be a linear refinement of  $R_{\mid V\setminus W}$.  Then
\[
L=L_1\cup L_2\cup \{(x,y)\in V^2:x\in W,\, y\in V\setminus W\}
\]
is a refinement of $R$ which is a linear order of $V$, that is, $L\in \mathbf{L}_\diamond(R)$. Moreover, $C_k(L)=\{W\}$.
 Hence $W\in C$.
\end{proof}

The following result is now immediate taking into account that, given a relation $R$ in $\mathbf{A}$, we have that its reversal is still in $\mathbf{A}$, and that $ |\max(R)|=|C_1(R)|$ and $|\min(R)|= |\max(R^r)|$.
\begin{corollary}\label{exist-max-min}
Let $R\in \mathbf{A}$.  Then  $1\leq |\max(R)|\leq |\mathbf{L}_\diamond(R)|$ and $1\leq|\min(R)|\leq |\mathbf{L}_\diamond(R)|.$
\end{corollary}

\begin{proposition}\label{inclusion}
Let $R\in \mathbf{A}$ and $k,l, m \in \ldbrack n-1 \rdbrack $ with $l\leq k\leq m$. Then
for every $W\in C_k(R)$, there exist $W'\in C_m(R)$ and $W''\in C_l(R)$ such that $W''\subseteq W\subseteq W'$.
\end{proposition}

\begin{proof} Let $W\in C_k(R)$. Thus $W\subseteq V$, $|W|=k$ and, for every $x\in W$ and $y\in V\setminus W$, we have $x\succeq_R y$.
We first show that there exists $W'\in C_m(R)$ such that $W\subseteq W'.$
The case $m=k$ is obvious. Assume then that $m>k$. Then $n\geq 3$ and  $1\leq m-k\le (n-k)-1$. Consider now the set $V\setminus W$, whose size is $n-k\ge 2$. The relation $R_{\mid V\setminus W}$ is quasi-acyclic and complete and so,  by Proposition \ref{ckR}, we have that $C_{m-k}(R_{\mid V\setminus W})\neq \varnothing.$ Pick $\hat{W}\in C_{m-k}(R_{\mid V\setminus W})$ and define $W'=W\cup \hat{W}$. We show that $W'\in C_m(R)$.
Surely $|W'|=m$ because the sets $W$ and $ \hat{W}$ are disjoint. Let $x\in W'$ and $y\in V\setminus W'$. Since $V\setminus W'\subseteq V\setminus W$, if $x\in W$ we have $x\succeq_R y$. Since $V\setminus W'\subseteq V\setminus \hat{W}$, if $x\in \hat{W}$ we also have $x\succeq_R y$.
We next show that there exists $W''\in C_l(R)$ such that $W''\subseteq W.$ The case $l=k$ is obvious. Assume then that $l<k$. Then $n\geq 3$ and $1\leq k-l\le k-1.$
Since $W$ has size $k\ge 2$ and the relation $R^r_{\mid W}$ is quasi-acyclic and complete,  by Proposition \ref{ckR} we have that $C_{k-l}(R^r_{\mid W})\neq \varnothing.$ Pick $\overline{W}\in C_{k-l}(R^r_{\mid W})$ and define $W''=W\setminus \overline{W}$. Then $W''\subseteq W\subseteq V$ and $|W''|=l$. We show that $W''\in C_l(R)$.
 Let $x\in W''$ and $y\in V\setminus W''.$ If $y\in V\setminus W$, then $x\succeq_R y$ because $x\in W$. If instead  $y\in \overline{W}$, since  $x\in W\setminus \overline{W}$, we have $y\succeq_{R^r}x$, that is, $x\succeq_R y$.
\end{proof}

\begin{proposition} \label{3}
Let $R\in\mathbf{A}$,  $k\in \ldbrack n-1 \rdbrack $ and $x^*\in V$. Assume that there exist distinct $y_1,\ldots, y_{n-k}\in V\setminus \{x^*\}$ such that,
for every $j\in \ldbrack n-k \rdbrack$, $x^*\succeq_R y_j$. Then there exists $W\in C_k(R)$ such that $x^*\in W$.
\end{proposition}

\begin{proof} By Proposition \ref{2}, there exists $L\in \mathbf{L}_{\diamond}(R)$ such that, for every $j\in \ldbrack n-k \rdbrack$, $x^*\succ_L y_j$. By Proposition \ref{ckR}, we have $C_k(R)\supseteq C_k(L)$ and we also know that $C_k(L)=\{W\}$ for a suitable $W\subseteq V$ with $|W|=k$.
Assume now by contradiction that $x^*\not\in W$. Then there exists $j^*\in  \ldbrack n-k \rdbrack$ such that $y_{j^*}\in W$. That leads to the contradiction $y_{j^*}\succ_L x^*$.
\end{proof}

We end the section with a result which describes the effects on the set of the $k$-maximum sets of reversing a complete and quasi-transitive relation $R$. It says that there exists at least a $k$-maximum set for $R$ which does not appear as a $k$-maximum set for the reversal of $R$.

\begin{proposition} \label{qA} Let $R\in \mathbf{T}$ with $R\neq V^2$. Then, for every $k\in  \ldbrack n-1 \rdbrack $,
$C_k(R)\not\subseteq C_k(R^r)$.
\end{proposition}

\begin{proof} Recalling that $\mathbf{A}\subseteq\mathbf{T}$,  by Corollary \ref{exist-max-min}, we have that $\max(R)\neq\varnothing$ and that $\min(R)\neq\varnothing$.
As a preliminary step, we show that there exist $x^*\in \max(R)$ and $y^*\in\min(R)$ such that $x^*\succ_R y^*$.

We first see that there exist $x^*\in \max(R)$ and $y\in V$ such that $x^*\succ_R y$. Assume, by contradiction, that that fact does not happen.
Since $R\neq V^2$ and $R$ is complete, there exist $x_0,x_1\in V$ such that $x_1\succ_R x_0$. In particular,  $x_1\neq x_0$. Since
 $x_1\notin\max(R)$, there exists $x_2\in V$ such that
$x_1\not\succeq_R x_2$.  By completeness,  we then have $x_2\succ_R x_1$ and, in particular, $x_2\neq x_1.$ We also have $x_2\neq x_0$, otherwise we would have the cycle $x_0\succ_R x_1\succ_R x_0.$
Since also $x_2\notin\max(R)$ we can repeat the argument constructing a sequence $x_0, \dots, x_m$ of arbitrary length $m\in\mathbb{N}$ of distinct elements of $V$,
against $V$ finite.

 We now show that there exists $y^*\in\min(R)$ such that $x^*\succ_R y^*$. Assume, by contradiction, that that fact does not happen. By what shown above, there exists $y_0\in V$ such that $x^*\succ_R y_0$. Since $y_0\not\in \min(R)$, there exists $y_1\in V$ such that $y_1\not\succeq_R y_0$. By completeness $y_0\succ_R y_1$ and, in particular,  $y_1\neq y_0$. Thus, by the transitivity of $S(R)$, we get $x^*\succ_R y_1$. Since $y_1\notin \min(R)$, there exists $y_2\in V$ with $y_1\succ_R y_2$. In particular  $y_2\neq y_1$. We also have $y_2\neq y_0$, otherwise we would have the cycle $y_0\succ_R y_1\succ_R y_0.$ Thus, as before, we can repeat the argument constructing a sequence $y_0, \dots, y_m$ of arbitrary length $m\in\mathbb{N}$ of distinct elements of $V$,
 against $V$ finite.

Fix now $x^*\in \max(R)$ and $y^*\in \min (R)$ such that $x^*\succ_R y^*$. Note that, in particular, $x^*\neq y^*$. We prove that there exists $W'\in C_k(R)$ such that $x^*\in W'$ and $y^*\not\in W'$. That implies the desired relation $C_k(R)\not\subseteq C_k(R^r)$ since certainly $W'\not\in C_k(R^r)$, due to the fact that $y^*\succ_{R^r} x^*$.

By Proposition \ref{3}, there exists $W\in C_k(R)$ such that $x^*\in W$. If $y^*\not \in W$, then set $W'=W$. Let then $y^*\in W$. Since $R_{|V\setminus W}$ is complete and acyclic, by Corollary \ref{exist-max-min}, there exists $z^*\in \max(R_{|V\setminus W})$.  Define then
$W'=(W\setminus \{y^*\})\cup\{z^*\}$ and show that it is indeed the $k$-subset we are looking for. Trivially we have that $|W'|=k$, $x^*\in W'$ and $y^*\not\in W'$. Thus we are left with proving that $W'\in C_k(R)$. Consider then $x\in W'$ and $y\in V\setminus W'$.
If $x\neq z^*$ and $y\neq y^*$, then $x\in W$ and $y\in V\setminus W$ so that $x\succeq_R y$. Let next $y=y^*$. Then $x\succeq_R y^*$ holds since $y^*\in \min(R)$.
Finally let  $x=z^*$ and $y\neq y^*$. Then $y\in V\setminus W$ and thus $z^*\succeq_R y$ because $z^*\in \max(R_{|V\setminus W})$.
\end{proof}

\section{Properties of the flow network rule}\label{fnr}

In this section we propose some properties of the flow network rule. Namely, we prove that that rule is decisive, neutral, homogeneous and that it satisfies non-imposition, efficiency, monotonicity, reversal symmetry and symmetry properties (Propositions \ref{BG-1}, \ref{neut-rule}, \ref{BG-2}, \ref{no-impo-rule}, \ref{BG-4}, \ref{BG-5}, \ref{BG-3} and \ref{sym-rule}, respectively). Moreover, we show that the flow network rule coincides with the Borda and the dual Borda network rule on a wide class of networks (Proposition \ref{balan-ut-rule}).
We emphasize that Propositions \ref{BG-2}, \ref{BG-4}, \ref{BG-5} and \ref{BG-3} are revisions of propositions stated, without proofs, in Belkin and Gvozdik (1989).

\begin{proposition}\label{BG-1}
Let $N\in\mathcal{N}$. Then $\mathfrak{F}_{\diamond}(N)\neq\varnothing$.
\end{proposition}

\begin{proof}
Apply Theorem \ref{F-qt} and Proposition \ref{lin-ref}.
\end{proof}

Given  $L\in \mathbf{L}$ and a bijection $\psi:V\rightarrow V$, let $\psi L\in \mathbf{L}$ be defined, for every $x,y\in V$, by setting $x\succeq_{\psi L}y$ if and only if $\psi^{-1}(x)\succeq_{L}\psi^{-1}(y).$

\begin{proposition}\label{neut-rule}
Let $N=(V,A,c) \in\mathcal{N}$ and $\psi:V\rightarrow V$ be a bijection. Then
$
\mathfrak{F}_\diamond(N^{\psi})=\{\psi L: L\in \mathfrak{F}_\diamond(N)\}.
$
\end{proposition}

\begin{proof} We first show that if $L\in \mathfrak{F}_\diamond(N)$ and $\psi:V\rightarrow V$ is a bijection, then $\psi L\in \mathfrak{F}_\diamond(N^{\psi})$.
Let $L\in \mathfrak{F}_\diamond(N)$. We need to see that, for every $x,y\in V$, $x\succeq_{\psi L}y$  implies $x\succeq_{N^{\psi}}y$. Fix then $x,y\in V$ and assume that
$x\succeq_{\psi L}y$, that is, $\psi^{-1}(x)\succeq_{L}\psi^{-1}(y)$. Since $L\in \mathfrak{F}_\diamond(N)$, we have $\psi^{-1}(x)\succeq_{N}\psi^{-1}(y)$ and, by Proposition \ref{neut-F},  we get $x\succeq_{N^{\psi}}y.$

Next let $M\in \mathfrak{F}_\diamond(N^{\psi})$ and prove that there exists $L\in \mathfrak{F}_\diamond(N)$ such that $M=\psi L$. Define then  $L=\psi^{-1}M$. Of course,
 we have that $M=\psi L$. Moreover, by the first part of the proof,  $L\in \mathfrak{F}_\diamond((N^{\psi})^{\psi^{-1}})=\mathfrak{F}_\diamond(N)$.
\end{proof}

By Proposition \ref{homo-net}, we immediately obtain the next result.

\begin{proposition}\label{BG-2}
Let $N\in\mathcal{N}$ and $\alpha\in\mathbb{N}$. Then $\mathfrak{F}_{\diamond}(N)=\mathfrak{F}_{\diamond}(\alpha N)$.
\end{proposition}

 \begin{proposition}\label{no-impo-rule}
Let $L\in\mathbf{L}$. Then there exists $N\in \mathcal{N}$ such that  $\mathfrak{F}_{\diamond}(N)=\{L\}$.
\end{proposition}
\begin{proof}  Consider the network $N_L$. Then, by Proposition \ref{main-qt}, we have $\mathfrak{F}(N_L)=L$ and thus $\mathfrak{F}_{\diamond}(N)=\{L\}.$
\end{proof}

\begin{proposition}\label{BG-4} Let $N=(V,A,c)\in \mathcal{N}$ and $x^*,y^*\in V$ with $x^*\neq y^*$. Assume that $c(x^*,y^*)\ge c(y^*,x^*)$ and, for every $z\in V\setminus\{x^*,y^*\}$, $c(x^*,z)\ge c(y^*,z)$ and $c(z,y^*)\ge c(z,x^*)$.
Then there exists $L\in \mathfrak{F}_{\diamond}(N)$ such that $x^*\succ_{L} y^*$.
Moreover, if one of the conditions $(a)$, $(b)$, $(c)$ and $(d)$ of Proposition \ref{Eff-F} holds true,
then, for every $L\in \mathfrak{F}_{\diamond}(N)$, $x^*\succ_{L} y^*$.
\end{proposition}

\begin{proof} By Proposition \ref{Eff-F}, we have that $x^*\succeq_N y^*$. Then, by Proposition \ref{2}, there exits $L\in \mathfrak{F}_\diamond(N)$ such that $x^*\succ_L y^*$.
Moreover, if one of the conditions $(a)$, $(b)$ ,$(c)$ and $(d)$ of Proposition \ref{Eff-F} holds true, Proposition \ref{Eff-F} gives $x^*\succ_N y^*$ and, by  Proposition \ref{2-bis}, we get $x^*\succ_{L} y^*$ for all $L\in \mathfrak{F}_{\diamond}(N)$.
\end{proof}

\begin{proposition}\label{BG-5}  Let $N=(V,A,c)\in \mathcal{N}$ and $x^* \in V$. Let $N'=(V,A,c')\in \mathcal{N}$ be such  that conditions $(a)$ and $(b)$ of Proposition \ref{Mon1-flow-rule} hold true.
Then, for every $L\in \mathfrak{F}_{\diamond}(N)$, there exists $L'\in \mathfrak{F}_{\diamond}(N')$ such that
$
\{y\in V: x^*\succ_L y\}\subseteq \{y\in V: x^*\succ_{L'}y\}.
$
Assume further that, for every $y\in V\setminus\{x^*\}$, $c'(x^*,y)> c(x^*,y)$ or $c'(y,x^*)< c(y,x^*)$. Then, for every $L\in \mathfrak{F}_{\diamond}(N)$ and  $L'\in \mathfrak{F}_{\diamond}(N')$, we have that
$
\{y\in V: x^*\succ_L y\}\subseteq \{y\in V: x^*\succ_{L'}y\}.
$
\end{proposition}

\begin{proof}
Let $L\in \mathfrak{F}_{\diamond}(N)$. Then
$
\{y\in V: x^*\succ_L y\}\subseteq \{y\in V\setminus\{x^*\}: x^*\succeq_N y\}.
$
By Proposition \ref{Mon1-flow-rule}, we have that
$
 \{y\in V\setminus\{x^*\}: x^*\succeq_N y\}\subseteq  \{y\in V\setminus\{x^*\}: x^*\succeq_{N'} y\}.
$
By Proposition \ref{2}, there exists $L'\in \mathfrak{F}_{\diamond}(N')$ such that
$
\{y\in V\setminus\{x^*\}: x^*\succeq_{N'} y\}\subseteq \{y\in V: x^*\succ_{L'} y\},
$
and the first part of the statement follows.

Assume now  that, for every $y\in V\setminus\{x^*\}$, $c'(x^*,y)> c(x^*,y)$ or $c'(y,x^*)< c(y,x^*)$ and consider $L\in \mathfrak{F}_{\diamond}(N)$ and $L'\in \mathfrak{F}_{\diamond}(N')$. Then
$
\{y\in V: x^*\succ_L y\}\subseteq \{y\in V\setminus\{x^*\}: x^*\succeq_N y\}.
$
By Proposition \ref{Mon1-flow-rule}, we have that
$
 \{y\in V\setminus\{x^*\}: x^*\succeq_N y\}\subseteq  \{y\in V\setminus\{x^*\}: x^*\succ_{N'} y\}.
$
By Proposition \ref{2-bis}, we have that
$
\{y\in V\setminus\{x^*\}: x^*\succ_{N'} y\}\subseteq \{y\in V: x^*\succ_{L'} y\},
$
and the second part of the statement follows.
\end{proof}

\begin{proposition}\label{BG-3}
Let $N\in\mathcal{N}$. Then $\mathfrak{F}_{\diamond}(N^r)=\{L^r: L\in \mathfrak{F}_{\diamond}(N)\}$.
\end{proposition}

\begin{proof}By Proposition \ref{Reversal}, we know that $\mathfrak{F}(N^r)=\mathfrak{F}(N)^r$. Then $\mathfrak{F}_{\diamond}(N^r)=\mathbf{L}_\diamond(\mathfrak{F}(N^r))=\mathbf{L}_\diamond(\mathfrak{F}(N)^r)$ and it is immediately observed that $\mathbf{L}_\diamond(\mathfrak{F}(N)^r)=\{L^r: L\in \mathfrak{F}_{\diamond}(N)\}$.
\end{proof}

\begin{proposition}\label{sym-rule}
Let $N\in\mathcal{N}$. Then $\mathfrak{F}_{\diamond}(N)=\mathbf{L}$ if and only if $N$ is pseudo-symmetric.
\end{proposition}

\begin{proof}
If $N$ is pseudo-symmetric then, by Proposition \ref{flat},  $\mathfrak{F}(N)=V^2$ so that $\mathfrak{F}_{\diamond}(N)=\mathbf{L}$. Assume now $\mathfrak{F}_{\diamond}(N)=\mathbf{L}$. We complete the proof showing that $\mathfrak{F}(N)=V^2$ and applying Proposition \ref{flat}. Consider $x,y\in V$. If $x=y$, then $x\succeq_N y$. If $x\neq y$ then suppose, by contradiction, that $x\not\succeq_N y$. Since $\mathfrak{F}(N)$ is complete, we have that
$y\succ_N x$. As a consequence, by Proposition \ref{2-bis}, we have that, for every $L\in\mathfrak{F}_{\diamond}(N)$, $y\succ_L x$ against $\mathfrak{F}_{\diamond}(N)=\mathbf{L}$.
\end{proof}

The Borda network rule $\mathfrak{B}_\diamond$ and the  dual Borda network rule $\widehat{\mathfrak{B}}_\diamond$ are respectively defined, for every $N\in\mathcal{N}$,
by
$\mathfrak{B}_\diamond(N)=\mathbf{L}_\diamond(\mathfrak{B}(N))$ and $\widehat{\mathfrak{B}}_\diamond(N)=\mathbf{L}_\diamond(\widehat{\mathfrak{B}}(N))$.
Immediately from Theorem \ref{OI}, we get the next proposition.
\begin{proposition}\label{balan-ut-rule}
Let $N\in \mathcal{B}\cup\mathcal{O}\cup\mathcal{I}$. Then $\mathfrak{F}_\diamond(N)=\mathfrak{B}_\diamond(N)=\widehat{\mathfrak{B}}_\diamond(N).$
\end{proposition}

\section{Properties of the flow $k$-multiwinner network solution}\label{fns}

This section is devoted to present some remarkable properties of the flow $k$-multiwinner network solution. In the five propositions below we prove that it is decisive, neutral, homogeneous and satisfies a non-imposition property as well as a consistency property with respect to the parameter $k$ (Propositions \ref{decisive}, \ref{neut-sol}, \ref{homo-net2}, \ref{no-impo-k} and \ref{inclusion2}, respectively).

\begin{proposition}\label{decisive}
Let $N\in\mathcal{N}$ and $k\in\ldbrack n-1 \rdbrack$. Then $\mathfrak{F}_k(N)\neq\varnothing$.
\end{proposition}

\begin{proof}
It immediately follows from Theorem \ref{F-qt} and Proposition \ref{ckR}.
\end{proof}

The next two results are easily obtained applying Proposition \ref{neut-F} and Proposition \ref{homo-net}.
\begin{proposition}\label{neut-sol}
Let  $N=(V,A,c) \in\mathcal{N}$, $\psi:V\rightarrow V$ bijective and $k\in \ldbrack n-1 \rdbrack $. Then
$
\mathfrak{F}_k(N^{\psi})=\{\psi(W): W\in \mathfrak{F}_k(N)\}.
$
\end{proposition}

\begin{proposition}\label{homo-net2}
Let  $N\in\mathcal{N}$, $\alpha\in\mathbb{N}$ and  $k\in\ldbrack n-1 \rdbrack$. Then $\mathfrak{F}_k(\alpha N)=\mathfrak{F}_k(N)$.
\end{proposition}

\begin{proposition}\label{no-impo-k}
Let $W\in \mathbb{P}_k(V) $.  Then there exists $N\in \mathcal{N}$ such that  $\mathfrak{F}_{k}(N)=\{W\}$.
\end{proposition}
\begin{proof}  Let $L$ be a linear oder having ranked in its first $k$ positions the element of $W.$  Then, by Proposition \ref{main-qt}, we have $\mathfrak{F}(N_L)=L$ and thus $\mathfrak{F}_{k}(N)=\{W\}.$
\end{proof}

\begin{proposition}\label{inclusion2}
Let  $N\in\mathcal{N}$ and $l,k,m \in\ldbrack n-1 \rdbrack$ with $l\le k\le  m$. For every $W\in \mathfrak{F}_k(N)$, there exists $W'\in \mathfrak{F}_m(N)$ and $W''\in \mathfrak{F}_l(N)$ such that $W''\subseteq W\subseteq W'$.
\end{proposition}

\begin{proof}
Apply Proposition \ref{inclusion} to the relation $\mathfrak{F}(N)$.
\end{proof}

The next two propositions show that $\mathfrak{F}_k$ satisfies properties sharing strong similarities with the concept of efficiency and monotonicity for social choice correspondences.

\begin{proposition}\label{Eff} Let $N=(V,A,c)\in \mathcal{N}$,  $x^*\in V$ and $k\in\ldbrack n-1 \rdbrack$. Assume that there exist distinct $y^*_1,\ldots, y^*_k\in V\setminus\{x^*\}$ such that,  for every $j\in \ldbrack k \rdbrack$ and $z\in V\setminus\{x^*,y^*_j\}$, we have $c(y^*_j,x^*)\ge c(x^*,y^*_j)$,
$c(y^*_j,z)\ge c(x^*,z)$ and $c(z,x^*)\ge c(z,y^*_j)$.
Then there exists $W\in \mathfrak{F}_k(N)$ such that  $x^*\notin W.$
Assume further that,  for every $j\in\ldbrack k \rdbrack$, at least one of the following conditions holds true:
\begin{itemize}
	\item[(a)] $c(y^*_j,x^*)> c(x^*,y^*_j)$;
	\item[(b)] for every $z\in V\setminus\{x^*,y^*_j\}$, $c(y^*_j,z)> c(x^*,z)$;
	\item[(c)] for every $z\in V\setminus\{x^*,y^*_j\}$,  $c(z,x^*)> c(z,y^*_j)$;
	\item[(d)] there exists $z\in V\setminus\{x^*,y^*_j\}$  such that $c(y^*_j,z)> c(x^*,z)$ and  $c(z,x^*)> c(z,y^*_j)$.
\end{itemize}
Then, for every $W\in \mathfrak{F}_k(N)$,  $x^*\notin W.$
\end{proposition}

\begin{proof}
By Proposition \ref{Eff-F}, we have that, for every $j\in\ldbrack k \rdbrack$, $y_j^*\succeq_N x^*$. Then, for every $j\in\ldbrack k \rdbrack$, $x^*\succeq_{N^r} y_j^*$. By Proposition \ref{3}, there exists $W'\in C_{n-k}(\mathfrak{F}(N^r))$ such that $x^*\in W'$.
Define now $W=V\setminus W'$. By Proposition \ref{Reversal},  we have $C_{n-k}(\mathfrak{F}(N^r))=C_{n-k}(\mathfrak{F}(N)^r),$
so that by Proposition \ref{ckrev}, we get  $W\in C_{k}(\mathfrak{F}(N))=\mathfrak{F}_k(N)$, with $x^*\not\in W$.

Assume  now that one among $(a)$, $(b)$, $(c)$ and $(d)$ holds true too. Then, by Proposition \ref{Eff-F}, for every $j\in\ldbrack k \rdbrack$, we have that  $y_j^*\succ_N x^*$ that is $x^*\succ_{N^r} y_j^*$. Thus, by Proposition \ref{3-bis}, we have $x^*\in W'$ for all $W'\in C_{n-k}(\mathfrak{F}(N^r))=C_{n-k}(\mathfrak{F}(N)^r)$.
By Proposition   \ref{ckrev}, it follows then that,  for every $W\in C_{k}(\mathfrak{F}(N))=\mathfrak{F}_k(N)$,  $x^*\notin W.$
\end{proof}

\begin{proposition}\label{Mon1} Let $N=(V,A,c)\in \mathcal{N}$, $k\in \ldbrack n-1 \rdbrack $, $W\in \mathfrak{F}_k(N)$ and $x^* \in W$. Let $N'=(V,A,c')\in \mathcal{N}$ be such that the conditions $(a)$ and $(b)$ of Proposition \ref{Mon1-flow-rule} hold true.
Then there exists $W'\in \mathfrak{F}_k(N')$ such that $x^*\in W'$.
Assume further that, for every $y\in V\setminus W$, we have $x^*\succ_N y$ or $c'(x^*,y)> c(x^*,y)$ or $c'(y,x^*)< c(y,x^*)$.
Then, for every $W'\in \mathfrak{F}_k(N')$, we have $x^*\in W'$.
\end{proposition}

\begin{proof} By Proposition \ref{Mon1-flow-rule} we know that, for every $y\in V\setminus W$, $x^*\succeq_{N'} y$. Then, applying Proposition \ref{3}, we conclude.
Assume now that, for every $y\in V\setminus W$, we have $x^*\succ_N y$ or $c'(x^*,y)> c(x^*,y)$ or $c'(y,x^*)< c(y,x^*)$. By Proposition \ref{Mon1-flow-rule} we know that, for every $y\in V\setminus W$, $x^*\succ_{N'} y$. Then applying Proposition
\ref{3-bis} we conclude.
\end{proof}

The next proposition describes the effects of reversing a network on the outcomes of the flow $k$-multiwinner network solution. Remarkably we have that if on a network  and on its reversal the flow $k$-multiwinner network solution selects a unique $k$-maximum set, those two sets cannot coincide.
We  refer to that property by saying that $\mathfrak{F}_k$ is immune to the reversal bias.  The concept of immunity to the reversal bias has been introduced by Saari and Barney (2003) for voting systems and recently studied by Bubboloni and Gori (2016a, 2016b) in the framework of social choice correspondences.

\begin{proposition}\label{5/2} Let $N\in \mathcal{N}$ and $k\in \ldbrack n-1 \rdbrack $. Then the following facts hold:
\begin{itemize}
	\item[(i)]$\mathfrak{F}_k(N)\neq \mathbb{P}_k(V)$ implies  $\mathfrak{F}_k(N)\not\subseteq\mathfrak{F}_k(N^r)$;
	\item[(ii)]$|\mathfrak{F}_k(N)|=1$ implies  $\mathfrak{F}_k(N)\not\subseteq \mathfrak{F}_k(N^r)$;
	\item[(iii)] if $\mathfrak{F}_k(N)=\{W\}$ and $\mathfrak{F}_k(N^r)=\{U\}$, then $W\neq U;$
         \item[(iv)] If $N=N^r$, then  $\mathfrak{F}_k(N)= \mathbb{P}_k(V)$;
	\item[(v)]$\mathfrak{F}_k(N)= \mathbb{P}_k(V)$ if and only if $\mathfrak{F}_k(N^r)= \mathbb{P}_k(V).$
\end{itemize}
\end{proposition}
\begin{proof}
$(i)$ Fix  $N\in \mathcal{N}$ with $\mathfrak{F}_k(N)\neq \mathbb{P}_k(V)$.  By Proposition \ref{totale}, we have $\mathfrak{F}(N)\neq V^2$. Since $\mathfrak{F}(N)\in \mathbf{T}$, by Propositions \ref{qA} and \ref{Reversal}, we get $\mathfrak{F}_k(N)=C_k(\mathfrak{F}(N))\not\subseteq C_k(\mathfrak{F}(N)^r)=C_k(\mathfrak{F}(N^r))=\mathfrak{F}_k(N^r).$

$(ii)$ Since $n\geq 2$ and $k\in\ldbrack n-1 \rdbrack$, we have that $|\mathbb{P}_k(V)|\geq 2.$ Thus $|\mathfrak{F}_k(N)|=1$ implies $\mathfrak{F}_k(N)\neq \mathbb{P}_k(V)$ and applying $(i)$, we get $\mathfrak{F}_k(N)\not\subseteq\mathfrak{F}_k(N^r).$

$(iii)$ It follows immediately from $(ii)$.

$(iv)$ Let $N=N^r$. Then, $\mathfrak{F}_k(N)=\mathfrak{F}_k(N^r)$ and $(i)$ gives $\mathfrak{F}_k(N)= \mathbb{P}_k(V)$.

$(v)$  By Proposition \ref{totale}, we have that $\mathfrak{F}_k(N)= \mathbb{P}_k(V)$ if and only if $\mathfrak{F}(N)=V^2$. But, obviously, $\mathfrak{F}(N)=V^2$ is also equivalent to $\mathfrak{F}(N)^r=V^2$ which, by Proposition \ref{Reversal}, is equivalent to $\mathfrak{F}(N^r)=V^2.$ This last fact,  by Proposition \ref{totale}, is in turn equivalent to $\mathfrak{F}_k(N^r)= \mathbb{P}_k(V).$
\end{proof}

The next proposition refers to the symmetry properties of the flow $k$-multiwinner network solution.

\begin{proposition}\label{sym-sol}
Let $N\in\mathcal{N}$ and $k\in \ldbrack n-1 \rdbrack $. Then $\mathfrak{F}_{k}(N)=\mathbb{P}_k(V)$ if and only if $N$ is pseudo-symmetric.
\end{proposition}

\begin{proof}
If $N$ is pseudo-symmetric then, by Proposition \ref{flat},  $\mathfrak{F}(N)=V^2$ so that $\mathfrak{F}_{k}(N)=\mathbb{P}_k(V)$. Assume now
$\mathfrak{F}_{k}(N)=\mathbb{P}_k(V)$. By Proposition \ref{totale} we get $\mathfrak{F}(N)=V^2$ and applying again Proposition \ref{flat}
we conclude that $N$ is pseudo-symmetric.
\end{proof}

Given $k\in\ldbrack n-1 \rdbrack$, the Borda $k$-multiwinner network solution $\mathfrak{B}_k$ and the  dual Borda $k$-multiwinner network solution $\widehat{\mathfrak{B}}_k$ are respectively defined, for every $N\in\mathcal{N}$,
by
$\mathfrak{B}_k(N)=C_k(\mathfrak{B}(N))$ and $\widehat{\mathfrak{B}}_k(N)=C_k(\widehat{\mathfrak{B}}(N)).$
Immediately, from Theorem \ref{OI}, we obtain that those classical solutions coincide with the $k$-multiwinner network solution on a wide class of networks.
\begin{proposition}\label{balan-ut-2}
Let $N\in \mathcal{B}\cup\mathcal{O}\cup\mathcal{I}$ and  $k\in \ldbrack n-1 \rdbrack$. Then $\mathfrak{F}_k(N)=\mathfrak{B}_k(N)=\widehat{\mathfrak{B}}_k(N).$
\end{proposition}

\section{The Schulze network method}\label{schulze-sec}

On the basis of the so called Schulze method introduced by Schulze (2011), we can formulate what we are going to call the {\it Schulze network method}.

Let $N=(V,A,c)\in\mathcal{N}$ and $x,y\in V$ with $x\neq y.$ Define, for every $\gamma=x_1\cdots x_m\in \Gamma(N,x,y)$,
\[
\delta (\gamma)=\mathrm{min}\{c(x_i,x_{i+1}): i\in\ldbrack m-1 \rdbrack\}
\]
 and put
\[
s_{xy}=\left\{
\begin{array}{ll}
\max\{\delta (\gamma):\gamma\in \Gamma(N,x,y)\}&\mbox{ if } \Gamma(N,x,y)\neq\varnothing\\
\vspace{-2mm}\\
0 &\mbox{ if } \Gamma(N,x,y)=\varnothing\\
\end{array}
\right.
\]
Let $s:A\rightarrow \mathbb{N}_0 $ be the map associating to every $(x,y)\in A$ the number $s_{xy}.$ We call $s$ the Schulze function.
Note that, due to the definition of path, $s_{xy}=0$ if and only if there exists no path from $x$ to $y$ in $N.$ Note also that $s_{xy}\geq c(x,y).$
Inspired by (2.2.5) in Schulze (2011), we get the following result.

\begin{proposition}\label{lemma-schulze}
Let $N\in\mathcal{N}$. Then $s$ satisfies the Gomory-Hu condition.
\end{proposition}

\begin{proof} We need to show that, for every $x,y,z\in V$ distinct, $s_{xz}\ge \min\{s_{xy},s_{yz}\}$. By contradiction, assume that there exist  $x,y,z\in V$ distinct such that $s_{xz}< \min\{s_{xy},s_{yz}\}$. Thus, $s_{xy}>s_{xz}$ and $s_{yz}>s_{xz}$. In particular, $s_{xy}>0$ and $s_{yz}>0.$

Thus, by definition of $s_{xy}$,  there exists a path $\gamma_{xy}=x_1\cdots x_m\in\Gamma(N,x,y) $, with $m\geq2$, such that $s_{xy}=\delta(\gamma_{xy})\ge \delta(\gamma)$ for all $\gamma\in \Gamma(N,x,y) $. In particular, for every $i\in \ldbrack m-1 \rdbrack$, we have
\begin{equation}\label{first -path}
c(x_i,x_{i+1})\ge s_{xy}>s_{xz}.
\end{equation}
Similarly, by definition of $s_{yz}$, there exists a path $\gamma_{yz}=y_1\cdots y_l\in\Gamma(N,y,z) $, with $l\geq 2$, such that $s_{yz}=\delta(\gamma_{yz})\ge \delta(\gamma)$ for all $\gamma\in \Gamma(N,y,z) $. In particular, for every $i\in \ldbrack l-1 \rdbrack$, we have
\begin{equation}\label{second -path}
c(y_i,y_{i+1})\ge s_{yz}>s_{xz}.
\end{equation}
Being $x_m=y=y_1$, the set $\{i\in \ldbrack m \rdbrack: x_i\in V(\gamma_{yz}) \}$ is nonempty and hence there exists $\nu=\mathrm{min}\{i\in \ldbrack m \rdbrack: x_i\in V(\gamma_{yz})\}.$
Since in a path the vertices are distinct, there exists a unique $\mu\in \ldbrack l \rdbrack$ such that $x_\nu=y_\mu.$  Moreover, due to $x\neq z$, we have $(\nu,\mu)\neq (1,l).$

 If $\mu=l$, then $x_\nu=z$ and $\nu\ge 2,$ so that, by \eqref{first -path},
 $\gamma=x_1\cdots x_\nu\in \Gamma(N,x,z)$ and $\delta (\gamma)>s_{xz}$ which contradicts the definition of $s_{xz}$. If instead $\mu\leq l-1$ then, by definition of $\nu$ and $\mu$,
the vertices $x_1,\dots, x_\nu,\,y_{\mu+1},\dots, y_l$ are all distinct and, by \eqref{first -path}  and \eqref{second -path}, all the arcs between two consecutive vertices in that list have capacity greater than $s_{xz}.$ It follows that
  $\gamma=x_1\cdots x_\nu\,y_{\mu+1}\cdots y_l\in \Gamma(N,x,z)$ and $\delta (\gamma)>s_{xz}$ which again contradicts the definition of $s_{xz}$.
\end{proof}

The Schulze network method $\mathfrak{S}$ is now defined associating with every $N\in \mathcal{N}$ the relation
$
\mathfrak{S}(N)=R(s).
$
Then, by Proposition \ref{lemma-schulze} and  Proposition \ref{schulze}, we immediately get the next result which is completely analogous to Theorem \ref{F-qt}.

\begin{theorem}\label{schulze-qt}
Let $N \in\mathcal{N}$. Then  $\mathfrak{S}(N)\in\mathbf{T}$.
\end{theorem}
Thus,  we can replicate,  word by word, all the definitions given in the flow environment. We obtain then the {\it Schulze network rule} defined, for every $N\in\mathcal{N}$, by  $\mathfrak{S}_{\diamond}(N)=\mathbf{L}_\diamond(\mathfrak{S}(N))$;
the {\it Schulze $k$-multiwinner network solution} defined, for every $N\in\mathcal{N}$, by $\mathfrak{S}_k(N)=C_k(\mathfrak{S}(N))$.
Of course, all the results obtained in the flow environment using only the fact that the flow relation is complete and quasi-transitive continue to hold true in the Schulze environment. In particular,  since it is immediately observed that, for every $N\in\mathcal{N}$, we have $\mathfrak{S}(N^r)=\mathfrak{S}(N)^r$, Proposition \ref{5/2} remains true substituting the symbol $\mathfrak{F}$ with $\mathfrak{S}$. As a consequence, for every $k\in\ldbrack n-1 \rdbrack$, $|\mathfrak{S}_k(N)|=1$ implies  $\mathfrak{S}_k(N)\not\subseteq \mathfrak{S}_k(N^r)$ so that, remarkably, the Schulze $k$-multiwinner network solution $\mathfrak{S}_k$ is immune to the reversal bias.

On the other hand, there are some crucial differences between the flow and the Schulze network methods.
Indeed, given $N\in\mathcal{N}$ and $x,y\in V$, considering $s_{xy}$ instead of $\varphi_{xy}$ to judge whether $x$ is at least as good as $y$ means to choose to focus only on the ``best path'' from $x$ to $y$ disregarding the contributes of the other paths. In our opinion, taking into account the whole network by considering the contribution not only of the main stream from $x$ to $y$ but also those of the secondary creeks,
is more appropriate in many contexts. The next example concretely illustrates our point of view.

\begin{example}\label{flow-schulze-1} {\rm Consider the competition $\textsf {D}$ among the three teams in $V=\{\textsc{a}, \textsc{b},\textsc{c}\}$  described by the table
 \[
\begin{array}{|c||c|c||c|c|}
\hline
\textsc{a} & 1 & 1  &\textsc{b}\\
\hline
\textsc{a} & 1 & 1 &\textsc{c}\\
\hline
\textsc{b} & 0 & 1 &\textsc{c}\\
\hline
\end{array}
\]
and let $N_{\textsf {D}}$ be the associated network. Then, for every $x\in \{\textsc{b},\textsc{c}\}$, we have $c(\textsc{a},x)=c(x,\textsc{a})=1$ and so also
$s_{\textsc{a}x}=s_{x\textsc{a}}=1$. Moreover $s_{\textsc{b} \textsc{c} }=s_{\textsc{c} \textsc{b} }=1,$ so that $s$ is a constant function. It follows that  $\mathfrak{S}(N_{\textsf {D}})=V^2$, so that, in particular, $\mathfrak{S}_1(N_{\textsf {D}})=\{\{\textsc{a}\}, \{\textsc{b}\},\{\textsc{c}\}\}.$
On the other hand, being $\varphi_{\textsc{a}\textsc{b}}=2$, $\varphi_{\textsc{b}\textsc{a}}=1$, $\varphi_{\textsc{a}\textsc{c}}=1$, $\varphi_{\textsc{c}\textsc{a}}=2$,
$\varphi_{\textsc{b}\textsc{c}}=1$, $\varphi_{\textsc{c}\textsc{b}}=2$,
we instead have that $\mathfrak{F}(N_{\textsf {D}})$ is the linear order $L$ given by $\textsc{c}\succ\textsc{a}\succ\textsc{b}$ and so $\mathfrak{F}_1(N_{\textsf {D}})=\{\{\textsc{c}\}\}.$\footnote {Note that, on the network $N_{\textsf {D}}$, the Kemeny network rule and the ranked pair network rule  are equal to $\{L, \,\textsc{c}\succ\textsc{b}\succ\textsc{a},\, \textsc{a}\succ\textsc{c}\succ\textsc{b}\}$, while $\mathfrak{F}_{\diamond}(N_{\textsf {D}})=\{L\}$. }
}
\end{example}
An accurate analysis  of the Schulze network method as well as a comparison with the flow network method is surely an interesting issue. It goes beyond the aim of the present paper.
\section{Comparison functions}\label{dutta-sec}

Following Dutta and Laslier (1999), a comparison function on $V$ is a function $g:V^2\to \mathbb{R}$ such that, for every $(x,y)\in V^2$, $g(x,y)=-g(y,x)$. Denote by $\mathcal{G}$ the set of all the comparison functions on $V$. Any correspondence from $\mathcal{G}$ to $V$ is called a choice correspondence\footnote{Actually, this definition is less general than the one in the paper of Dutta and Laslier. Nevertheless, it is sufficient for our purposes.}. Some interesting choice correspondences are defined by Dutta and Laslier (1999) in their paper. They are called the uncovered set ($UC$), the sign-uncovered set ($SUC$), the minimal covering set ($MC$), the sign minimal covering set
($SMC$), the essential set ($ES$) and the sign essential set ($SES$) and, as their names suggest, can be seen as extensions of some classical tournament solutions.

Given $N=(V,A,c)\in\mathcal{N}$, it can be naturally associated with it the comparison function $g_N$, defined as follows
\[
g_N(x,y)=
\left\{
\begin{array}{ll}
c(x,y)-c(y,x) & \mbox{ if }x\neq y\\
0&\mbox{ if }x=y
\end{array}
\right.
\]
The function $g_N$ is called the margin function of $N$. As a consequence, any choice correspondence $F$ induces the 1-multiwinner network
solution $\mathfrak{G}_F$ associating with any $N\in\mathcal{N}$ the set $\{\{x\}\subseteq V: x\in F(g_N)\}$.
Thus, it can be interesting a comparison between the flow 1-multiwinner network
solution  and the 1-multiwinner network solutions induced by the choice correspondences above mentioned. Even though a deep analysis of that problem it is outside the purposes of this paper, we can easily prove that, in general,
$\mathfrak{F}_1\neq \mathfrak{G}_{UC}$ and $ \mathfrak{F}_1\neq \mathfrak{G}_{MC}$.
Indeed, consider $N=(V,A,c)$ where
\[
\begin{array}{l}
V=\{\textsc{a},\textsc{b},\textsc{c},\textsc{d}\};\\
\vspace{-2mm}\\
A=\{(x,y)\in V^2: x,y\in\{\textsc{a},\textsc{b},\textsc{c},\textsc{d}\} \mbox{ and } x\neq y\};\\
\vspace{-2mm}\\
c:A\to\mathbb{N}_0 \mbox{ is defined by }\\
c(\textsc{a},\textsc{b})=2,\; c(\textsc{b},\textsc{a})=0,\;c(\textsc{a},\textsc{c})=2,\; c(\textsc{c},\textsc{a})=0,\\
c(\textsc{a},\textsc{d})=1,\; c(\textsc{d},\textsc{a})=0,\;  c(\textsc{b},\textsc{c})=2,\;c(\textsc{c},\textsc{b})=0,\\
c(\textsc{b},\textsc{d})=1,\; c(\textsc{d},\textsc{b})=0,\;c(\textsc{c},\textsc{d})=1,\; c(\textsc{d},\textsc{c})=0.\\
\end{array}
\]
A simple computation shows that $\{\textsc{d}\}\not\in \mathfrak{F}_1(N)$. Moreover,
$g_N$ is such that
\[
\begin{array}{l}
g_N(\textsc{a},\textsc{b})=2,\; g_N(\textsc{b},\textsc{a})=-2,\;g_N(\textsc{a},\textsc{c})=2,\; g_N(\textsc{c},\textsc{a})=-2,\\
g_N(\textsc{a},\textsc{d})=1,\; g_N(\textsc{d},\textsc{a})=-1,\;  g_N(\textsc{b},\textsc{c})=2,\;g_N(\textsc{c},\textsc{b})=-2,\\
g_N(\textsc{b},\textsc{d})=1,\; g_N(\textsc{d},\textsc{b})=-1,\;g_N(\textsc{c},\textsc{d})=1,\; g_N(\textsc{d},\textsc{c})=-1.\\
\end{array}
\]
Since $g_N$ is the same comparison function considered in Example 4.1 in Dutta and Laslier (1999), we get that $\{\textsc{d}\}\in UC(g_N)=MC(g_N)$ so that
$\mathfrak{F}_1(N)\neq \mathfrak{G}_{UC}(N)$ and $ \mathfrak{F}_1(N)\neq \mathfrak{G}_{MC}(N)$.

\vspace{10mm}
\noindent {\Large{\bf References}}
\vspace{2mm}

\noindent Aziz, H., Brill, M., Fischer, F., Harrenstein, P., Lang, J., Seeding, H.G., 2015. Possible and necessary winners of partial tournaments. Journal of Artificial Intelligence Research 54, 493-534.
\vspace{2mm}

\noindent Bang-Jensen, J., Gutin, G., 2008. {\it Digraphs Theory, Algorithms and Applications}. Springer.
\vspace{2mm}

\noindent Belkin, A.R., Gvozdik. A.A., 1989. The flow model for ranking objects. Problems of Cybernetics, Decision Making and Analysis of Expert Information, Eds. Dorofeyuk, A., Litvak, B., Tjurin, Yu., Moscow, 109-117 (in Russian).
\vspace{2mm}

\noindent Bouyssou, D., 1992. Ranking methods based on valued preference relations: A characterization of the net flow method. European Journal of Operational Research
60, 61-67.\vspace{2mm}

\noindent Brandt, F., Brill, M., Harrenstein, P., 2016. Extending tournament solutions. Proceedings of the Twenty-Eighth AAAI Conference on Artificial Intelligence.
\vspace{2mm}

\noindent Bubboloni, D., Gori, M., 2015. Symmetric majority rules. Mathematical Social Sciences 76, 73-86.
\vspace{2mm}

\noindent Bubboloni, D., Gori, M., 2016a. On the reversal bias of the Minimax
social choice correspondence. Mathematical Social Sciences  81, 53-61.
 \vspace{2mm}

\noindent Bubboloni, D., Gori, M., 2016b. Resolute refinements of social choices correspondences. Mathematical Social Sciences 84, 37-49.
 \vspace{2mm}

\noindent Csat\'o, L., 2017. Some impossibilities of ranking in generalized tournaments. arXiv:1701.06539v4.
\vspace{2mm}

\noindent De Donder, P., Le Breton, M., Truchon, M.,  2000. Choosing from a weighted tournament. Mathematical Social Sciences 40, 85-109.
\vspace{2mm}

\noindent Dutta and Laslier, 1999. Comparison functions and choice correspondences. Social Choice and Welfare 16, 513-532.
 \vspace{2mm}

\noindent Elkind, E., Faliszewski, P., Skowron, P., Slinko, A., 2017. Properties of Multiwinner Voting Rules. Social Choice and Welfare. doi:10.1007/s00355-017-1026-z.
 \vspace{2mm}

\noindent  Gomory, R. E., Hu T. C., 1961. Multi-terminal network flows.
Journal of the Society for Industrial and Applied Mathematics 9, No. 4, 551-570.
\vspace{2mm}

\noindent Gonz\'alez-D\'iaz, J., Hendrickx, R., Lohmann, E., 2014.
Paired comparisons analysis: an axiomatic approach of ranking methods. Social Choice and Welfare 42, 139-169.
\vspace{2mm}

\noindent Gvozdik, A.A., 1987. Flow interpretation of the tournament model for ranking. Abstracts of the VI-th Moscow Conference of young scientists on cybernetics and computing, Moscow: Scientific Council on Cybernetics of RAS, 1987, p.56 (in Russian).
\vspace{2mm}

\noindent Hartmann, M., Schneider, M.H., 1993. Flow symmetry and algebraic flows. ZOR - Methods and Models of Operations Research 38, 261-267.
\vspace{2mm}

\noindent Kitti, M., 2016. Axioms for centrality scoring with principal eigenverctors. Social Choice and Welfare 46, 639-653.
\vspace{2mm}

\noindent Langville, A. N., Meyer, C. D., 2012. Who Is $\#1$? The Science of Rating and Ranking. Princeton University Press.
\vspace{2mm}

\noindent Laslier, J.-F., 1997. {\it Tournament solutions and majority voting}. Studies in Economic Theory, Volume 7. Springer.
\vspace{2mm}

\noindent Lov\'asz, L., 1973. Connectivity in digraphs. Journal of Combinatorial Theory (B) 15, 174-177.
\vspace{2mm}

\noindent Palacios-Huerta, I., Volij, O., 2004. The measurement of intellectual influence. Econometrica 72, 963-977.
\vspace{2mm}

\noindent Patel, Jaymin, 2015. College football rankings: maximum flow model. Industrial Engineering Undergraduate Honors Theses, paper 35.
\vspace{2mm}

\noindent Peris, J.E., Subiza, B., 1999. Condorcet choice correspondences for weak tournaments. Social Choice and Welfare 16, 217-231.
\vspace{2mm}

\noindent Saari, D.G., 1994. {\it Geometry of Voting}. In: Studies in Economic Theory, vol. 3. Springer.
\vspace{2mm}

\noindent Saari, D.G., Barney, S., 2003. Consequences of reversing preferences. The Mathematical Intelligencer 25, 17-31.
\vspace{2mm}

\noindent Schulze, M., 2011.   A new monotonic, clone-independent, reversal symmetric, and Condorcet-consistent single-winner election method. Social Choice and Welfare  36, 267-303.
\vspace{2mm}

\noindent Slutzki, G., Volij, O., 2005. Ranking participants in generalized tournaments. International Journal of Game Theory 33, 255-270.
\vspace{2mm}

\noindent Szpilrajn, E., 1930. Sur l'extension de l'ordre partiel.  Fundamenta Mathematicae 16, 386-389.
\vspace{2mm}

\noindent van den Brink, R., Gilles, R.P., 2009. The outflow ranking method for
weighted directed graphs. European Journal of Operational Research 193, 484-491.
\vspace{2mm}

\noindent Vaziri, B., Dabadghao, S., Yih, Y., Morin, T. L., 2017. Properties of sports ranking methods. Journal of the Operational Research Society, 1-12, doi:10.1057/s41274-017-0266-8.

\end{document}